\documentclass[journal]{IEEEtran}

\usepackage[utf8]{inputenc}
\usepackage{enumitem}
\usepackage{graphicx,tikz,subfigure,epsfig,psfrag,color,epstopdf,soul,enumitem}
\usepackage{amsfonts,amsmath,amssymb,bm,stmaryrd}
\usepackage{nth,csquotes}
\usepackage{mathtools}
\usepackage[ruled,lined]{algorithm2e}
\usepackage{cite}

\newcommand{\argmin}[2]{\smash{\mathop{{\rm arg\,min}}\limits_{#1}}\ #2 }
\newcommand{\ts}[1]{\bm{\mathcal{#1}}}
\graphicspath{{fig//}}

\newenvironment{proof}[1][Proof]{\begin{trivlist}
\item[\hskip \labelsep {\bfseries #1}]}{\end{trivlist}}

\newlength{\textlarg}

\definecolor{green}{RGB}{34,139,34}

\makeatletter
\def\rank{\mathop{\operator@font rank}\nolimits}
\def\krank{\mathop{\operator@font krank}\nolimits}
\def\vecd{\mathop{\operator@font vecd}\nolimits}
\def\diag{\mathop{\operator@font diag}\nolimits}
\def\cov{\mathop{\operator@font cov}\nolimits}
\def\imagin{\mathop{\operator@font Im}\nolimits}

\makeatother
\newcommand{\eqdef}{\stackrel{\mathrm{def}}{=}} 
\renewcommand{\t}{{\mathsf T}} 
\newcommand{\h}{{\mathsf H}} 
\newcommand{\kron}{\boxtimes}  
\newcommand{\con}{\mathop{\bullet}\limits} 
\newcommand{\out}{\mathop{\otimes}\limits}  

\newcommand{\va}{\mathbf{a}}

\newcommand{\vb}{\mathbf{b}}

\newcommand{\vc}{\mathbf{c}}
\newcommand{\vd}{\mathbf{d}}

\newcommand{\vy}{\mathbf{y}}

\newcommand{\mA}{\mathbf{A}}

\newcommand{\mD}{\mathbf{D}}

\newcommand{\mM}{\mathbf{M}}

\newcommand{\mQ}{\mathbf{Q}}

\newcommand{\mT}{\mathbf{T}}
\newcommand{\mU}{\mathbf{U}}
\newcommand{\mSigm}{\bm{\Sigma}}
\newcommand{\mV}{\mathbf{V}}

\newcommand{\mY}{\mathbf{Y}}

\newcommand{\fin}{\color{black}}
\newcommand{\red}{\color{black}}

\newcommand{\soul}{\color{black}}
\newcommand{\mv}[1]{\mathbf{#1}}
\newcommand{\mr}[1]{\mathrm{#1}}

\newcommand{\vv}{\mathbf{v}}

\newcommand{\argtimes}[2]{\con_{#1}^{#2}}

\newtheorem{theorem}{Theorem}
\newtheorem{corollary}{Corollary}
\newtheorem{experiment}{Experiment}

\newtheorem{lemma}{Lemma}

\definecolor{myblue}{rgb}{.8, .8, 1}
  
\definecolor{myorange}{rgb}{0.9961,    0.6250,    0.4766}
   
\newlength{\largeur}
\begin{document}

\title{A Simultaneous Sparse Approximation Method for Multidimensional Harmonic Retrieval}

\author{Souleymen~Sahnoun*,
        El-Hadi~Djermoune,~\IEEEmembership{Member,~IEEE},
        David~Brie,~\IEEEmembership{Member,~IEEE},
        Pierre Comon,~\IEEEmembership{Fellow,~IEEE}%
    \thanks{This work is funded by the European Research Council under the European Community's Seventh Framework Programme FP7/2007--2013 Grant Agreement no.~320594, DECODA project.}%
    \thanks{S. Sahnoun and P. Comon are with the Gipsa-Lab, 11 rue des Math\'ematiques, Domaine Universitaire BP 46, 38402 Saint Martin d'H\`eres Cedex, France (e-mail: firstname.name@gipsa-lab.grenoble-inp.fr).}%
    \thanks{E.-H. Djermoune and D. Brie are with the Centre de Recherche en Automatique de Nancy, Universit\'e de Lorraine, CNRS, Boulevard des Aiguillettes, BP 239, 54506 Vandoeuvre Cedex, France (phone: +33(0)383 68 44 72; fax: +33(0)383 68 44 61; e-mail: firstname.name@univ-lorraine.fr).}%
    \thanks{*Corresponding author.}%
}



\maketitle

\begin{abstract}
In this paper, a \soul sparse-based \fin method for the estimation of the parameters of multidimensional ($R$-D) modal (harmonic or damped) complex signals in noise is presented. The problem is formulated as $R$ simultaneous sparse approximations of multiple 1-D signals. To get a method able to handle large size signals while maintaining a sufficient resolution, a multigrid dictionary refinement technique is associated with the simultaneous sparse approximation problem.  The refinement procedure is proved to converge in the single $R$-D mode case. Then, for the general multiple modes $R$-D case, the signal tensor model is decomposed in order to handle each mode separately in an iterative scheme. The proposed method does not require an association step since the estimated modes are automatically ``paired''. We also derive the Cram\'er-Rao lower bounds of the parameters of modal $R$-D signals. The expressions are given in compact form in the single $R$-D mode case. Finally, numerical simulations are conducted to demonstrate the effectiveness of the proposed method.
\end{abstract}


 \ifCLASSOPTIONpeerreview
 \begin{center} \bfseries EDICS Category: SSP-SPEC, SSP-PARE \end{center}
 \fi
%
\IEEEpeerreviewmaketitle

\section{Introduction}\label{sec:intro}

\IEEEPARstart{T}{he} problem of estimating the parameters of sinusoidal signals from noisy measurements is an important topic in signal processing and several parametric and nonparameteric approaches have been developed for one-dimensional (1-D) signals~\cite{Stoica:97}. Recently, this problem has received a renewed interest thanks to the emergence of multidimensional ($R$-D) applications. Indeed, parameter estimation from $R$-D signals is required in numerous applications in signal processing and communications such as nuclear magnetic resonance (NMR) spectroscopy, wireless communication channel estimation~\cite{gershman2005space} and MIMO radar imaging~\cite{Nion:10}. In all these applications, signals are assumed to be a superposition of $R$-D sinusoids or, more generally, of exponentially decaying $R$-D complex exponentials (modal signals). As for the 1-D case, the crucial step is the estimation of the $R$-D modes (including frequencies and damping factors) because they are nonlinear functions of the data.

In order to achieve high resolution estimates, parametric approaches are often preferred to nonparamatric ones. Several parametric $R$-D methods ($R\geq 2$) have been proposed. They include linear prediction-based methods such as 2-D TLS-Prony~\cite{Sacchini93}, and subspace approaches such as matrix enhancement and matrix pencil (MEMP)~\cite{Hua92}, 2-D ESPRIT~\cite{Rouquette2001}, multidimensional folding (MDF)~\cite{mokios20043d}, improved multidimensional folding (IMDF)~\cite{Lui2006,Lui07_november}, Tensor-ESPRIT~\cite{haardt2008higher}, principal-singular-vector utilization for modal analysis (PUMA)~\cite{sun2012accurate,so2010efficient} and the methods proposed in \cite{huang2012multidimensional,lin2013efficient}. All these methods perform at various degrees but it is generally admitted that they yield accurate estimates at high SNR scenarios and/or when the frequencies are well separated. This is obtained at the expense of computational effort. For instance, MDF, IMDF, 2-D ESPRIT and MEMP are ESPRIT-type techniques. They require to build large size matrices and apply the ESPRIT-based method, which make their computational complexity very high particularly in the case of large $R$-D signals. The Tensor-ESPRIT algorithm uses the structure inherent in the $R$-D data at the expense of a high computational complexity. 
Recently, TPUMA~\cite{sun2012accurate} was proposed as an accurate and computationally efficient multidimensional harmonic retrieval method, which attains the Cram\'er-Rao lower bound (CRLB) and does not require to build large size matrix or tensor. However its performance degrades rapidly with the increase of the number of components present in the $R$-D signal.



 \fin 
 Recently, methods based on sparse approximations have been proposed to address the harmonic or modal retrieval problem~\cite{Sahnoun:2012Eusipco,goodwin1999matching,Malioutov2005,Stoica:2011spice1,sahnoun2012_Mltigrid,Sahnoun:13,sward2014high,adalbjornsson2014high}. \fin 
 For time-data spectral estimation, the dictionary is formed from a set of (normalized) complex exponentials potentially embedded in the data, which allows one to easily include some prior knowledge about the position of certain known modes. More generally, the usual choice is a uniform spectral grid obtained by sampling the frequency and damping factor lines. Clearly, a fine grid will lead to a good resolution but, on the other hand, it will result in a huge dictionary \cite{Sahnoun:2012Eusipco}. This complexity is further increased in the case of $R$-D signals in which we are confronted with $2R$-D grids.

\soul In order to reduce the computational burden, a multigrid scheme for sparse approximation was proposed in \cite{sahnoun2012_Mltigrid}  to iteratively refine the dictionary starting from a coarse one. At each iteration, a sparse approximation is performed and then new grid points (atoms) are inserted in the vicinity of active ones leading to a multiresolution-like scheme. \soul So the multigrid algorithm in \cite{sahnoun2012_Mltigrid} refines jointly $R~2$-D grids. This algorithm is efficient but has mainly two drawbacks: 1) it does not have convergence guarantees, 2) the dictionary becomes intractable for large signals when $R\geq 2$.

\fin

The goal of the present paper is to propose a fast multidimensional modal estimation technique able to handle large signals and yielding a good estimation accuracy. 
\begin{enumerate}[leftmargin = .5cm]
\item First, the proposed approach, as for some parametric methods for modal retrieval, is based on the idea of estimating the parameters independently along each dimension $r=1,\ldots,R$. It will be shown that the \emph{simultaneous} sparse approximation concept~\cite{Tropp:06,Sahnoun:13} is well-suited for $R$-D modal retrieval ($R\geq 2$).  
\item \soul The second contribution consists in the proposition of a new multigrid scheme which amounts to consider a  two-step refinement of 1-D grids, the first step for frequencies and the second one for damping factors. One advantage of the two-step multigrid is that it reduces the computational time. The convergence issue of the proposed multigrid strategy is analyzed firstly for single tone $(F=1)$ case, and convergence conditions are derived.  Condition for convergence are expressed in terms of atom positions in the initial dictionaries. 
\item The extension of this result to the multiple tones case ($F>1$) is not trivial because, not only it depends on the selected sparse approximation algorithm, but also on the coherence of the dictionary~\cite{Tropp:06}. Indeed, due to the refinement strategy, the resulting dictionary is far from being uncorrelated which may prevent convergence even in the noiseless case. Consequently, for the multiple tones case, we exploit an alternative representation of the data model enabling the extraction of the $R$-D signal tones separately. Therefore, the third contribution of this paper consists in deriving a new algorithm for estimating parameters of $R$-D damped signals in which the results of the previous contribution apply. \red The effectiveness of the new algorithm for multiple $R$-D tones is also analyzed. \soul One very interesting by-product of this approach is that the pairing of $R$-D parameters  is achieved for free, without any further association stage. 
\end{enumerate}
 \fin

To assess the performances of an estimation method, the usual way consists in comparing the variance of the estimates to the CRLB. \soul In \cite{Hua92} Y. Hua derived the CRLB for 2-D frequencies, i.e., undamped 2-D exponentials; no damped signals are considered. Closed-form expressions of the CRLB for the general undamped $R$-D case are derived in \cite{boyer2008deterministic}. CRLB for 2-D damped signals are derived in \cite{Clark1994}. Therefore, \fin to the best of our knowledge, no compact expressions of the CRLB's are available for the general $R$-D modal (damped) signal. Thus, another contribution of this paper is the derivation of the CRLB's for the frequency, damping factor, amplitude and phase of the $R$-D modal signal. 

The remainder of this paper is organized as follows. In section~\ref{sec:problem_statment}, we introduce notation and present the $R$-D modal retrieval problem. In section~\ref{sec:simultaneous}, we formulate the $R$-D modal estimation problem as $R$ simultaneous sparse estimation problems, show how to construct a modal dictionary on a uniform grid and then \soul describe the new fast multigrid  strategy. \fin In section~\ref{sec:singleRD}, we give sufficient conditions for convergence of the multigrid dictionary refinement scheme in the case of single tone $R$-D signals. In light of these new results, we propose in section~\ref{sec:multipleRD} a new efficient algorithm for multiple tones $R$-D modal signals. In section~\ref{sec:crlb}, we derive the expressions of the CRLB's for the parameters of $R$-D damped exponentials in Gaussian white noise. We then give the CRLB in the cases of single damped  and  undamped $R$-D cisoids. The effectiveness of the proposed method is demonstrated using simulation signals in section~\ref{sec:simul}. Finally, conclusions are drawn in section~\ref{sec:conclusion}.

\section{Notation and Problem Statement}\label{sec:problem_statment}

\subsection{Notation}
In this paper, scalars are denoted as lower-case letters $(a,b,\alpha)$, column vectors as lower-case bold-face letters $(\mv{a},\mv{b})$, matrices as bold-face capitals $(\mv{A},\mv{B})$, and tensors as calligraphic bold-face letters $(\bm{\mathcal A}, \bm{\mathcal B})$. Let $(\cdot)^\t,(\cdot)^\h$ and $(\cdot)^{\dagger}$ denote the transpose, the Hermitian transpose and the pseudo-inverse, respectively.  The symbols \enquote{$\odot$} and  \enquote{$\kron$}  will denote the Khatri-Rao product (column-wise Kronecker)  and the Kronecker product, respectively.  Both words ``mode'' and ``tone'' are used to refer to a component of the multidimensional signal. 
 The tensor operations used here are consistent with~\cite{Como14:spmag}:
\begin{enumerate}[label = $\bullet$, leftmargin = .5cm]
\item the outer product of two tensors $\bm{\mathcal A} \in{\mathbb{C}^{M_1\times \cdots \times M_R}} $ and $\bm{\mathcal B} \in{\mathbb{C}^{K_1\times \cdots \times K_N}} $ is given by:
\begin{align}
	&\bm{\mathcal C} = \bm{\mathcal A}  \out  \bm{\mathcal B} \in{\mathbb{C}^{M_1\times \cdots \times M_R \times K_1\times \cdots \times K_N}}, \nonumber\\
	& c(m_1,\ldots, m_{R},k_1,\ldots ,k_N) = \nonumber \\ 
	&\hspace{1.8cm} a(m_1,\ldots, m_{R})b(k_1,\ldots ,k_N)
\end{align}
\item the contraction product acting on the $r$th index of a tensor  $\bm{\mathcal A} \in{\mathbb{C}^{M_1\times \cdots \times M_R}} $ and the $2$nd index of a matrix $\mv{U} \in \mathbb{C}^{K\times M_r}$ is:
 \begin{align}
 & \bm{\mathcal B} = \bm{\mathcal A} \con_r \mv{U} \in {\mathbb{C}^{M_1\times \cdots \times M_{r-1} \times K\times M_{r+1} \times \cdots \times M_R}}, \nonumber \\
 	& b(m_1,m_2,\ldots, m_{r-1},k_r,m_{r+1},\ldots ,m_R)\nonumber =\\ & \hspace{1.5cm}\sum_{m_r=1}^{M_r} a(m_1,m_2,\ldots, m_{R})u(k_r,m_r)
 \end{align}
 \item the matrix $\mv{A}_{(r)} \in \mathbb{C}^{M_r\times (M_1\cdots M_{r-1}M_{r+1}\cdots M_R)}$ represents the unfolding (dimension-$r$ matricization) of the tensor $\bm{\mathcal A}$ and corresponds to the arrangement of the dimension-$r$ fibers of $\bm{\mathcal A}$ to be the columns of the resulting matrix.
 \item \soul $\|\ts{A}\|^2$ denotes the the Frobenius norm for tensors: $\| \ts{A}\|^2 = \sum_{m_1,\ldots,m_R} |a({m_1,\ldots,m_R})|^2$.   \fin
\end{enumerate}

\subsection{Problem Formulation}

An $R$-D modal signal is modeled as the superposition of $F$ multidimensional damped complex sinusoids:
\begin{equation}
\tilde{y}(m_1, \ldots, m_R) = \sum_{f=1}^{F}c_f \prod_{r=1}^{R}a_{f,r}^{m_r-1}+e(m_1, \ldots, m_R) \label{eq:rd_signal_model}
\end{equation}
where $m_r=1,\ldots, M_r$ for $r=1,\ldots,R$. $M_r$ denotes the sample support of the $r$th dimension, $a_{f,r}=\exp{(\alpha_{f,r}+j \omega_{f,r})} \in \mathbb{C}$ is the $f$th mode in the $r$th dimension, $\{ \alpha_{f,r}\}_{f=1, r=1}^{F,R}$, $\alpha_{f,r}\in\mathbb{R}^-$, are the damping factors, $\{ \omega_{f,r}=2\pi\nu_{f,r}\}_{f=1, r=1}^{F,R}$ are the angular frequencies, and $c_f=\lambda_f\exp(j\phi_f)$ is the complex amplitude of the $f$th mode where $\lambda_f=|c_f|$ denotes the magnitude and $\phi_f$ the phase. 
 $e(m_1,m_2,\ldots,m_R)$ is a zero-mean complex Gaussian white noise with variance $\sigma^2$ and mutually independent components in all dimensions. Throughout this paper, the tilde symbol (~$\tilde{}$~) denotes a noisy signal; \emph{e.g.} $\tilde{y}(\cdot)=y(\cdot)+e(\cdot)$.

In a tensor form, the $R$-D signal in (\ref{eq:rd_signal_model}) may be written as
\begin{equation}
{\ts{\widetilde{Y}}}= \bm{\mathcal{Y}}+ \bm{\mathcal{E}} \label{eq:tensor_form}
\end{equation}
where $\{\ts{\widetilde{Y}}, \bm{\mathcal{Y}}, \ts{E}\} \in \mathbb{C}^{M_1\times M_2 \times\cdots\times M_R}$. The problem consists in estimating the set of parameters $\{a_{f,r}\}_{f=1,r=1}^{F,R}$ and $\{c_{f}\}_{f=1}^{F}$ from the $R$-D signal samples.

\section{Simultaneous Sparse Approximation for $R$-D Modal Signals}\label{sec:simultaneous}

\subsection{Tensor Formulation of the Data Model}

The noise-free data tensor $\ts{Y}$ in (\ref{eq:tensor_form}) can be written in the following form:
\begin{equation}\label{eq:Y-a-i-r}
\bm{\mathcal{Y}} = \sum\limits_{f=1}^{F} c_f~\mv{a}_{f,1} \out \mv{a}_{f,2} \out \cdots \out\mv{a}_{f,R}
\end{equation}
where  $\va_{f,r}=[1, a_{f,r}, \ldots, a_{f,r}^{M_r-1}]^\t$, $r=1,\ldots,R$.  Equation (\ref{eq:Y-a-i-r}) is called the Canonical Polyadic (CP) decomposition form,  or the Candecomp/Parafac decompostion  of the tensor $\ts{Y}$~\cite{Como14:spmag,Kolda:2009tensor}. The CP model (\ref{eq:Y-a-i-r}) can be concisely denoted by
\begin{align}
\ts{Y} = \llbracket \mv{c}; \mv{A}_1, \mv{A}_2, \ldots, \mv{A}_R \rrbracket
\end{align}
where $\mv{A}_r= [\mv{a}_{1,r},\mv{a}_{2,r}, \ldots, \mv{a}_{F,r} ]$, $r=1,\ldots,R$, and $\mv{c}=[c_1,c_2,\ldots, c_F]^\t$ is the vector of complex amplitudes. Using these definitions, the matricized form of $\ts{Y}$ along the $r$th dimension is given by
\begin{align}
\mv{Y}_{(r)}= \mv{A}_{r} \bm{\Delta}_{\mv{c}} (\mv{A}_{R} \odot \cdots \odot \mv{A}_{r+1} \odot \mv{A}_{r-1} \odot \cdots \odot \mv{A}_{1})^\t
\end{align}
where $\bm{\Delta}_{\mv{c}} = \text{diag}(\mv{c})$. Then, we can write
\begin{equation}\label{eq:unfoldingYtilde}
\mv{\widetilde{Y}}_{(r)}= \mv{A}_{r} \mv{H}_{r}+ \mv{E}_{(r)}
\end{equation}
where $\mv{H}_{r} \in \mathbb{C}^{F\times M_{r}'}$ is
\begin{equation}
\mv{H}_{r} \eqdef \bm{\Delta}_{\mv{c}} (\mv{A}_{R} \odot \cdots \odot \mv{A}_{r+1} \odot \mv{A}_{r-1} \odot \cdots \odot \mv{A}_{1})^\t
\end{equation}
and $M_{r}' = \prod\limits_{\begin{smallmatrix}k=1\\ k \neq r\end{smallmatrix}}^R M_k$. Therefore
\begin{align}
\mv{Y}_{(r)} & \eqdef [\mv{y}_{(r),1},\ldots, \mv{y}_{(r),M_{r}'}] \nonumber \\
	&= \left[ \sum\limits_{f=1}^F h_r(f,1)\mv{a}_{f,r}, \ldots, \sum\limits_{f=1}^F h_r(f,M_r')\mv{a}_{f,r}\right] \label{eq:Yr-him}
\end{align}
where $h_r(f,m_r')$ is the $(f,m_r')$ entry of the matrix $\mv{H}_{r}$, for $f=1,\ldots,F$ and $m_r'=1,\ldots,M_r'$. We observe that, for a given $r$, the columns $\mv{y}_{(r),m_r'}$ of $\mv{Y}_{(r)}$ are linear combinations of the vectors $\{\mv{a}_{f,r}\}_{f=1}^F$. Hence, the columns $\mv{y}_{(r),m_r'}$ can be considered as multiple experiences involving the same one-dimensional signal generated by the modes $a_{f,r}, f = 1, \ldots, F,$ but with different amplitudes for each experience. This property will be used in the next section to formulate the problem of estimating the mode coordinates in the $r$th dimension as a simultaneous sparse approximation problem.

\vspace{-3mm}
\subsection{Simultaneous Sparse Approximation}\label{subsec:simul_sparse}

\soul Assuming\footnote{Note that this assumption is considered only in this section.} that $M_r> F,\forall r$, it is easy to see from (\ref{eq:Yr-him}) that for a fixed $r$ the mode coordinates $\{a_{f,r}\}_{f=1}^{F_r}$ ($F_r\leq F$) in the $r$th dimension are identifiable from any column of $\mv{Y}_{(r)}$. \fin This process can also be repeated on each dimension $r=1,\ldots,R$ to get all the modes coordinates. In practice, we have to replace the matrix $\mv{Y}_{(r)}$ by its noisy counterpart $\mv{\widetilde{Y}}_{(r)}$ accounting for the additive white noise. In this case, (\ref{eq:Yr-him}) holds only approximately. Consequently, for each column $\tilde{\mv{y}}_{(r),m_r'}, m_r'=1,\ldots,M_r'$, the modal estimation problem can be formulated as a sparse approximation problem corresponding to the following constrained optimization:
\begin{equation}
\mv{x}_{m_r'}= \arg \min_{\mv{x}} \| \mv{x}\|_0 \quad\text{subject to}\quad \|\tilde{\mv{y}}_{(r),m_r'}-\mv{Q}_r\mv{x} \|_{2}^2\leq \epsilon
\end{equation}
where $\mv{Q}_r\in\mathbb{C}^{M_r\times N}$, $N\gg M_r$, is a (known) modal dictionary, $\mv{x}\in\mathbb{C}^{N}$ is a (sparse) vector containing the coefficients of the activated columns in $\mv{Q}_r$, and $\epsilon$ is a small reconstruction error related to the noise variance. The pseudo-norm $\|\mv{x}\|_0$ counts the number of nonzero elements in a vector $\mv{x}$. The design of $\mv{Q}_r$ is discussed in section~\ref{sec:dictionary}. The fact that each vector  $\tilde{\mv{y}}_{(r),m_r'}$  corresponds to a 1-D signal generated by the same modes implies that the position of nonzero entries in $\mv{x}_{m_r'}$ should be the same for $m_r'=1,2,\ldots,M_R'$. Let $\mv{X}$ be the matrix defined by
\begin{equation}
\mv{X}= [\mv{x}_1, \mv{x}_2, \ldots, \mv{x}_{M_r'}],
\end{equation}
then the sparsity of $\mv{X}$ may be measured by computing the Euclidian norms of the rows; those providing a nonzero norm define the rows of the activated atoms (which are estimations of modes $a_{f,r}$ in the dimension $r$) in the dictionary $\mv{Q}_r$. Therefore, we are facing a simultaneous sparse approximation problem:
\begin{equation}\label{eq:argminX}
\mv{X}_r = \arg \min_{\mv{X}} \| \mv{X}\|_{2,0} \quad\text {subject to}\quad \|\widetilde{\mv{Y}}_{(r)}-\mv{Q}_r\mv{X} \|_{F}^2\leq \epsilon
\end{equation}
where
\begin{align}
\| \widetilde{\mv{Y}}_{(r)}  -\mv{Q}_r\mv{X} \|_{F}^2 &= \|\mathrm{vec}( \widetilde{\mv{Y}}_{(r)}  -\mv{Q}_r\mv{X}) \|_{2}^2,\\
\| \mv{X}\|_{2,0} &= \left\| \begin{bmatrix} \|\mv{X}^{1,:}\|_2& \cdots &\|\mv{X}^{N,:}\|_2 \end{bmatrix}^\t \right\|_0,
\end{align}
and $\mv{X}^{n,:}$ stands for the $n$th row of $\mv{X}$. The simultaneous sparse representation models, called also Multiple Measurement Vectors (MMV), have been studied from several angles of view, and different approaches have been proposed~\cite{Rakotomamonjy2011surveying}, using greedy strategies~\cite{Tropp:04} such as Simultaneous Orthogonal Matching Pursuit (SOMP)~\cite{Tropp:06}, convex relaxation methods~\cite{Tropp:16}, randomized algorithms such as REduce MMV and BOost (ReMBo)~\cite{Mishali:08} and subspace-augmented MUSIC~\cite{Lee:12}. 
As the goal of the present paper is to develop a fast approach well adapted to large signals, we restrict our attention to the SOMP algorithm~\cite{Tropp:06}, \soul reported in Appendix~\ref{appendix_SOMP}. \fin However, it is worth mentioning that, in more intricate cases and/or small size signals, much more efficient simultaneous sparse algorithms may be used at the price of an increased computational burden. A straightforward way to get the $R$-tuples $\{(a_{f,1},\ldots,a_{f,R})\}_{f=1}^F$ consists in estimating the modes ${a_{f,r}}$ in the $R$ dimensions using matrices $\mv{\widetilde{Y}}_{(r)}, r= 1, \ldots, R$, which requires a further pairing step to form the $R$-D modes in the multiple tones case ($F >1$). To get accurate estimates using the described scheme, two conditions have to be satisfied, 1) the dictionary should contain all possible modes present in the signal, 2)  the sparse approximation method should have sufficient guarantees for selecting the true atoms from the dictionary, which is known as ``exact recovery guarantees''. These problems are discussed in the following sections and an alternative representation of the data is used to avoid the pairing stage in the multiple tones case.

\subsection{Modal Dictionary Design and Multigrid Strategy} \label{sec:dictionary}
\subsubsection{Uniform Modal Dictionary}

The dictionary $\mv{Q}_r\in\mathbb{C}^{M_r\times N}$ can be defined as follows. Let $N_\mu$ be the number of points of a uniform grid covering the frequency interval $[0,1)$. Similarly, let $N_\beta$ be the number of points of a uniform grid covering the damping factor interval $(\beta_{\min}, 0]$, where $\beta_{\min}$ is a lower bound on $\{\alpha_{f,r}\}_{f=1}^{F}$. Then $\mv{Q}_r$ is given by
\begin{align}
	\mv{Q}_r =[ &\mv{q}_r(0,0), \ldots, \mv{q}_r((N_\mu-1)\delta_{\mu}, 0),  \mv{q}_r(0, \delta_{\beta}), \ldots, \nonumber \\ &\mv{q}_r((N_\mu-1)\delta_{\mu}, \delta_{\beta}), \ldots, \mv{q}_r((N_\mu-1)\delta_{\mu},(N_\beta-1)\delta_{\beta})]
\end{align}
where $\mv{q}_r(\mu, \beta)=\mv{a}_r(\mu, \beta)/||\mv{a}_r(\mu,\beta)||_2$ with $\mv{a}_r(\mu, \beta)= [1, e^{(\beta+j 2 \pi \mu)}, \ldots, e^{(\beta+j 2 \pi \mu)(M_r-1)}]^\t$, 
$\delta_{\beta}= \beta_{\min}/N_\beta$, and $\delta_{\mu}=1/N_\mu$. In short, $\mv{Q}_r$ is obtained from a discretization of the $(\nu,\alpha) $ plane. Each point of the grid corresponds to a hypothetic mode. The total number of columns in $\mv{Q}_r$ is $N=N_\mu N_\beta\gg F$, each of them is called atom. In the aim of reducing the computational complexity, we propose to estimate frequencies and then damping factors by calling twice the sparse approximation method. At the first step, the frequencies are estimated using a harmonic dictionary. In the second step, the damping factors are estimated using a modal dictionary formed by the already estimated frequencies and a damping factor grid. These two steps are explained in section~\ref{sec:singleRD}.

\subsubsection{Multi-Grid Dictionary Refinement} \label{subsec:multigrid}
To achieve a high-resolution modal estimation, a possible way is to define uniform grids as before and selecting very small values for $\delta_\mu$ and $\delta_\beta$ to retrieve the frequencies and damping factors, respectively. As a consequence, the resulting dictionaries will lead to prohibitive calculation cost and memory capacities requested. Rather, we propose to start with a coarse one ($N_\mu$ and $N_\beta$ low) and to adaptively refine it through a multigrid scheme. The procedure is the same for estimating the frequency and damping factor. The principle is sketched on Figure~\ref{fig_multi}. The main idea is the adaptation of the dictionary as a function of the previous dictionary and the estimated coefficients. Let $\ell$ be the current grid level ($\ell=0,\ldots,L-1$). At level $\ell$, we first restore the signal $\mathbf{X}_r(\ell)$ related to the dictionary $\mathbf{Q}_r(\ell)$ by applying the SOMP method. Then we refine the dictionary by inserting atoms inbetween pairs of $\mathbf{Q}_r(\ell)$, in the neighborhood of each activated atom and we apply again the SOMP method at level $\ell+1$ to restore $\mathbf{X}_r(\ell+1)$ with respect to the refined dictionary $\mathbf{Q}_r(\ell+1)$. This process is repeated until a desired level of resolution is reached. Algorithm~\ref{algo:dec_refined} presents the one-step dictionary refinement (DICREF), from level $\ell$ to $\ell+1$, where, for $a$ and $b$ reals, $\mr{linspace}(a,b,\eta)$ generates a  set of $\eta$ equispaced points  in the interval $[a,b]$. 
\soul The difference between the present framework and that in \cite{sahnoun2012_Mltigrid} is the following. In  \cite{sahnoun2012_Mltigrid} the multigrid algorithm refines jointly $R~2$-D grids, which leads to expensive computations when $R\geq 2$, without convergence guarantees. The present mutigrid scheme refines linear grids, which leads to low computational complexity with convergence guarantees as we will show in the next section.  \fin

Finding the convergence conditions of the \soul new \fin multigrid strategy in the general case (multiple tones) is not easy and depends on the selected sparse approximation algorithm. By contrast, it is possible to show that, under mild conditions, the convergence may be guaranteed in the single tone case. This issue is discussed in the next section. In section~\ref{sec:multipleRD}, we make use of an alternative representation of the data model in the case of multiple tones and a method allowing one to retrieve the signal tones separately.


\begin{figure}[t]
\centering
%
%
%
%

\begin{tikzpicture}[xscale=1.2,yscale=1.1]
\draw[->] (0,0) -- (0,3.5) node[above] {\footnotesize Level};
\draw[->,gray] (0,0) -- (5,0) node[below,black] {\footnotesize $\mu$ or $\beta$};
\draw[->,gray] (0,1.5) -- (5,1.5) node[below,black] {\footnotesize $\mu$ or $\beta$};
\draw[->,gray] (0,2.5) -- (5,2.5) node[below,black] {\footnotesize $\mu$ or $\beta$};
\draw[very thick,dotted] (2.5,.5) -- (2.5,1);
\draw[very thick,dotted] (2.5,3) -- (2.5,3.5);

\foreach \y/\ytext in {0,1.5/$\ell$,2.5/$\ell+1$}
 \draw (-.1,\y) node[left] {\footnotesize \ytext};
 
\foreach \x in {0,1,...,4}
 \foreach \y in {0,1.5,2.5}
  \draw (\x,\y) circle (0.06cm);
\foreach \x in {.75, 1.25,2.5,3.5} {
 \draw (\x,1.5) circle (0.06cm);
 \draw (\x,2.5) circle (0.06cm);
}

\foreach \x in {1,2,3} \fill (\x,0) circle (0.03cm);
\foreach \x in {.5,1.5,2.75} {
 \draw (\x,1.5) circle (0.06cm);
 \fill (\x,1.5) circle (0.03cm);
 \draw (\x,2.5) circle (0.06cm);
}

\foreach \x in {.25,.625,1.375,1.75,2.625,2.875} \draw[red] (\x,2.5cm+.1pt) node {\footnotesize $\boldsymbol\varoast$};
\end{tikzpicture}
\caption{The multigrid dictionary refinement procedure with $\eta=1$. ($\varbigcirc$) atoms in the dictionary; ($\mathrlap{\hspace*{2.2pt}\bullet}\varbigcirc$) activated atoms; (\textcolor{red}{$\boldsymbol\varoast$}) new atoms}
\label{fig_multi}
\end{figure}
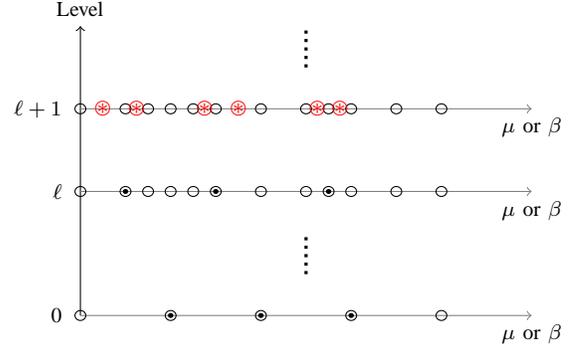

\IncMargin{0mm}\begin{algorithm}[t]
\caption{ Dictionary refinement $(\mr{DICREF})$} \label{algo:dictionary_refinement}
\footnotesize
\DontPrintSemicolon
\SetKwData{Left}{left}\SetKwData{This}{this}\SetKwData{Up}{up} \SetKwFunction{Union}{Union}
\SetKwFunction{FindCompress}{FindCompress} \SetKwInOut{Input}{input}\SetKwInOut{Output}{output}
\Input {A vector $\vd \in \mathbb{R}^N$ of sorted frequencies or damping factors, an index set $\Omega$ of activated atoms, the  number of atoms $\eta \in \mathbb{N}$ to add at each side of an activated one}
\Output{Updated vector $\vd_{\mr{updated}}$}
\BlankLine
	\For{$i=1:\rm{numel}(\Omega) $}{
		$\vd_{i,1} =  \mr{linspace}\left(\vd(\Omega(i)-1), \vd(\Omega(i)), \eta \right)$\; 
		$\vd_{i,2}= \mr{linspace}(\vd(\Omega(i)), \vd(\Omega(i)+1), \eta )$\; 
		$\vd_{i} = [\vd_{i,1}^\t, \vd_{i,2}^\t(2:\eta) ]^\t$\; 
		}
	$\vd_{\mr{updated}} =  \mr{union}  (\vd_1, \ldots, \vd_{\mr{numel}(\Omega)})$ \;
	\Return{ $\vd_{\mr{updated}}$}
\label{algo:dec_refined}
\end{algorithm}

\section{Single $R$-D Mode Estimation} \label{sec:singleRD}


In the previous section, we have shown how the $R$-D modal retrieval problem may be tackled using a sparse approximation algorithm by estimating the set of parameters in each dimension $r=1,\ldots,R$. Here, we give the sufficient conditions for convergence of the multigrid dictionary refinement scheme for $F=1$. Without loss of generality, we set $R=1$. For notation simplicity, we omit reference to the dimension index $r$.

According to~(\ref{eq:rd_signal_model}), the 1-D modal signal containing a single mode can be written as follows:  
\begin{equation}
y(m) = c_1 a_{1}^{m-1} =  c_1 e^{(\alpha_1+j2\pi\nu_1)(m-1)}, m=1,\ldots,M.
\end{equation}
Let $\mQ$ be a normalized modal dictionary $\mQ=[\mathbf{q}_1,\ldots,\mathbf{q}_N]$, with
\begin{equation}
\mathbf{q}_n=\frac{1}{\sqrt{\sum_{m}|q_n|^{2m}}} [1,q_n,\ldots,q_n^{M-1}]^\t,
\end{equation}
$q_n=\exp(\beta_n+j2\pi\mu_n), \mu_n\in[0,1), \beta_n\in (\beta_{\min},0]$, for $n=1,\ldots,N$. The single tone sparse approximation of $\mathbf{y}$ with respect to $\mQ$ is the solution of the criterion:
\begin{equation}
\min_{\mathbf{x}} J(\mathbf{x})=||\mathbf{y}-\mQ\mathbf{x}||^2 \quad \text{s.t.} \quad ||\mathbf{x}||_0=1.
\end{equation}
The optimal solution is given by
\begin{equation}
x^*_n=\mathbf{q}_n^\h\mathbf{y}, \quad \mathbf{x}^*_{\{1,\ldots,N\}\backslash n}=\mathbf{0}, \quad J(\mathbf{x}^*)=||\mathbf{y}||^2-\mathbf{y}^\h\mathbf{q}_n\mathbf{q}_n^\h\mathbf{y}
\end{equation}
where $n$ is the selected column number in $\mQ$. Finally, the minimum $J(\mathbf{x}^*)$ is reached for an atom $\mathbf{q}_n$ that maximizes $J'(\mathbf{q}_n)=\mathbf{y}^\h\mathbf{q}_n\mathbf{q}_n^\h\mathbf{y}=|\mathbf{q}_n^\h\mathbf{y}|^2$, $n=1,\ldots,N$.

\subsection{Estimating the Frequency: The Harmonic Dictionary}  
First, we estimate frequency $\nu_1$ using a harmonic dictionary (i.e. assuming $\beta_n=0,\forall n$). In this case, we have:
\begin{equation}\label{eq:Jprimqn}
J'(\mu_n) = \frac{|c_1|^2}{M}\left|\frac{1-e^{\alpha_1M+j2\pi (\nu_1-\mu_n)M}}{1-e^{\alpha_1+j2\pi (\nu_1-\mu_n)}}\right|^2.
\end{equation}
The following theorem gives a sufficient condition for the multigrid dictionary refinement scheme to converge to the global maximum of $J'$.

\begin{theorem}
Let $y(m)$ be a single tone ($F=1$) noiseless signal of length $M$ and $\mathbf{Q}(\ell=0)=[\mathbf{q}_1~\mathbf{q}_2~\ldots~\mathbf{q}_{N(0)}]^\t$ be the initial harmonic dictionary in which the columns are sorted in increasing order of $\mu_n(0), n=1,2,\ldots,N(0)$ and covering the frequency interval [0,1): $\mu_1(0)=0$ and $\mu_{N(0)}(0)=1-1/M$. Then the refinement scheme is convergent (i.e. $\exists n\in\{1,\ldots,N(\ell)\}$ s.t. $\lim_{\ell\to\infty}\mu_n(\ell)=\nu_1$) if the following condition is satisfied:
\begin{equation}
\max_{n\in\{1,\ldots,N(0)-1\}} |\mu_{n+1}(0)-\mu_{n}(0)|<2\zeta_M
\end{equation}
where $\zeta_M$ is a constant depending only on $M$.
\end{theorem}

\begin{proof}
It is easy to check that the global maximum of $J'(\mu_n)$ is reached for $\mu_n=\nu_1$, $\forall \alpha_1$. Figure~\ref{fig:conv1} shows the variation of $J'(\mu_n)$ as a function of $\mu_n$ for $\nu_1=0.1$, $\beta_n=0$ and $M=10$. For $\alpha_1=0$, $J'(\mu_n)$ reduces to a Fej\'er kernel which has exactly one local maximum in the interval $[\nu_1+k/M,\nu_1+(k+1)/M]$, $k\neq 0$. Let $J'_1$ be the maximum value of $J'(\mu_n)$ in the interval $[\nu_1+1/M,\nu_1+2/M]$ and $\nu_1+\zeta_M$ be the value of $\mu_n$ such that $J'(\mu_n=\nu_1+\zeta_M)=J'_1$ in the interval $[\nu_1,\nu_1+1/M]$ (we assume\footnote{The case of $M\leq2$ in not of practical interest but the theorem is still valid by setting $\zeta_M=\frac{1}{2}$ because $J'(\mu_n)$ is a monotonically decreasing function in the interval $[\nu_1,\min\{\frac{1}{2},\frac{1}{2}+\nu_1\}]$.} that $M>2$). For the dictionary refinement strategy to converge to the global maximum, it is sufficient to the sparse approximation algorithm to select, at a given level $\ell$, an atom whose frequency satisfies $|\mu_{n^*}(\ell)-\nu_1|<\zeta_M< 1/M$, where $\mu_{n^*}(\ell)=\arg\max_n J'(\mu_n)$. Indeed, if $\mu_{n^*}(\ell)\in(\nu_1-\zeta_M,\nu_1+\zeta_M)$ then adding two atoms whose frequencies are located on both sides of $\mu_{n^*}(\ell)$ will lead to the selection, at level $\ell+1$, of an atom that satisfies $|\mu_{n^*}(\ell+1)-\nu_1|\leq |\mu_{n^*}(\ell)-\nu_1|$: the distance between the selected atom and the true frequency is a monotonically decreasing sequence. Finally, the convergence is guaranteed if the initial dictionary contains  an atom $n$  such that $|\mu_n(0)-\nu_1|<\zeta_M$, which is satisfied if
\begin{equation}\label{eq:cond}
\max_{n\in\{1,\ldots,N(0)-1\}} |\mu_{n+1}(0)-\mu_{n}(0)|<2\zeta_M.
\end{equation}
given the fact that the sequence $\{\mu_n(0)\}$ covers the interval $[0,1)$. For $\alpha_1<0$, the main lobe of $J'(\mu_n)$ becomes broader and $\zeta_M$ larger than for $\alpha_1=0$. Consequently, condition~(\ref{eq:cond}) is also sufficient for $\alpha_1<0$.
\hfill $\blacksquare$
\end{proof}

\begin{corollary}
In the single tone case, the harmonic dictionary refinement is convergent if the initial frequency grid ($\ell=0$) is the Fourier grid.
\end{corollary}
\begin{proof}
Fourier bins are obtained for $N=M$ and $\mu_n(0)=(n-1)/M$. Since $\zeta_M>1/2M$, the proof is straightforward because $|\mu_{n+1}(0)-\mu_{n}(0)|=1/M<2\zeta_M$.
 \hfill $\blacksquare$
\end{proof}

It is important to note that condition~(\ref{eq:cond}) is sufficient but not necessary. Moreover, this condition is established when adding a single atom on both sides of the selected one (i.e. $\eta=1$ in Algorithm~\ref{algo:dictionary_refinement}). When $\eta\gg 1$, the condition may be relaxed and the rate of convergence is expected to be higher.

\begin{figure}[t]
\centering
\psfrag{f}[cc][cb]{\tiny $\mu_n$}
\psfrag{J}[cc][lc]{\tiny $J'(\mu_n)/||\mathbf{y}||^2$}
\psfrag{zeta}[cc][cc]{\tiny $\zeta_M$}
\includegraphics[width=\largeur]{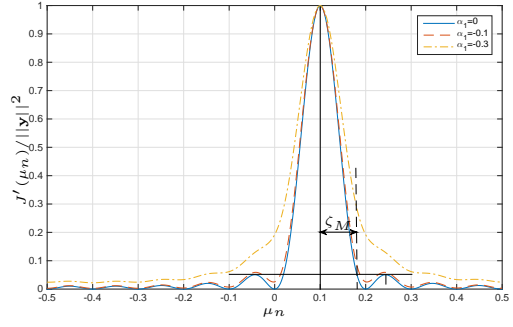}
\caption{$J'(\mu_n)$ in the single mode case with $\nu_1=0.1$ and $\beta_n=0$}
\label{fig:conv1}
\end{figure}

\subsection{Estimating the Damping Factor: The Modal Dictionary} 
Assume that the previous sparse approximation method using a harmonic dictionary has converged to \soul select \fin an atom with $\mu_n=\nu_1$. Now, we have to estimate the damping factor $\alpha_1$. We form a modal dictionary using the damping factor grid and the frequency $\nu_1$, i.e. $q_n = \exp(\beta_n + j2\pi \nu_1)$. Consequently,
\begin{align}\label{eq:Jdamp}
J'(\beta_n) 
 &= \frac{|c_1|^2(1-e^{2\beta_n})}{1-e^{2\beta_n M}}\left(\frac{1-e^{(\alpha_1+\beta_n) M}}{1-e^{(\alpha_1+\beta_n) }}\right)^2.
\end{align}
%
%

\begin{theorem}
\label{theorem:dampMiltigrid}
Let $y(m)$ be a single tone ($F=1$) noiseless signal of length $M$ and $\mathbf{Q}(0)=[\mathbf{q}_1~\mathbf{q}_2~\ldots~\mathbf{q}_{N(0)}]^\t$ be the initial modal dictionary formed using the frequency $\nu_1$, i.e., $q_n = \exp(\beta_n(0) + j2\pi\nu_1)$, where $\nu_1$ is the frequency of signal $\vy$. The columns are sorted in increasing order of $\beta_n(0), n=1,2,\ldots,N$ and covering the damping factor interval $(\beta_{\min}, 0]$. Then the refinement scheme is convergent (i.e. $\exists n$ s.t. $\lim_{\ell\to\infty}\beta_n(\ell)=\alpha_1$) if $\alpha_1 \in (\beta_{\min}, 0]$.
\end{theorem}
\begin{proof}
Let $g(\beta_n)$ be the derivative of $J'(\beta_n) $ in~(\ref{eq:Jdamp}) with respect to $\beta_n$. It is easy to check that $g(\beta_n)>0$ for $\beta_n<\alpha_1$, $g(\beta_n)<0$ for $\beta_n>\alpha_1$, and $g(\beta_n)=0$ when $\beta_n=\alpha_1$.
In other words, $J'(\beta_n) $ is monotonically increasing before the maximum reached at $\alpha_1$ and  monotonically decreasing after $\alpha_1$. Therefore, the multigrid algorithm converges to $\alpha_1$ if $\beta_{\min} < \alpha_1$.  
\hfill $\blacksquare$
\end{proof}
As a consequence of Theorem~\ref{theorem:dampMiltigrid}, the initial modal dictionary can be formed using only two points in the damping factor grid: $\beta_1(0) = \beta_{\min}$ and $\beta_2(0) = 0$. 

We can now state that the multigrid algorithm based on two sparse approximations (for frequency and then damping factor) converges in the single tone case under some conditions. \soul Note that in the noisy case when the SNR is sufficiently high, the convergence analysis  is still valid as in the noiseless case, and the proposed multigrid sparse scheme for single tone converges to the global maximum of the Fejér kernel. \fin The extension to the single tone $R$-D modal retrieval problem is straightforward and can be performed  according  to the formulation presented in  Section~\ref{subsec:simul_sparse}. The details of this approach (\emph{STSM: Single Tone Sparse Method}) are presented in Algorithm~\ref{algo:STSM}. The algorithm takes as input a noisy single tone $R$-D signal, and a couple of integers $\eta_\nu$ and $\eta_\alpha$ that correspond respectively to the number of frequency and damping factor atoms to be added on both sides of the corresponding selected ones.  Next, for each dimension $r=1,\ldots,R$, we execute  two tasks to estimate the frequency and then the damping factor: in each step we apply SOMP combined to DICREF algorithm using corresponding dictionaries and taking into account the convergence conditions discussed previously. Then parameters of $a_r$, i.e., $\nu_r$ and $\alpha_r$, are given by the corresponding selected atoms. This approach will be exploited in the next section for the multiple tones case.

\IncMargin{0mm}\begin{algorithm}[t]
\caption{Sparse multigrid method for single tone estimation ($\mr{STSM}$)}
\footnotesize
\DontPrintSemicolon
\SetKwData{Left}{left}\SetKwData{This}{this}\SetKwData{Up}{up} \SetKwFunction{Union}{Union}
\SetKwFunction{FindCompress}{FindCompress} \SetKwInOut{Input}{input}\SetKwInOut{Output}{output}
\Input {A tensor $\ts{Y} \in \mathbb{C}^{M_1\times \cdots \times M_R}$, $(\eta_\nu, \eta_\alpha)\in \mathbb{N}\times\mathbb{N}$ }
\Output{Parameters of the single $R$-D mode: $a_1, \ldots,a_R$}
\BlankLine
{\bf initialization}: $(k_\nu, k_\alpha)=(0,0)$\;
initialize $\vd_\nu^{(0)}$ and $\vd_\alpha^{(0)}$ using $\zeta$ \;
\BlankLine
\For{$r=1:R $} {
	\While{halting criterion false}{
		$k_{\nu}=k_{\nu}+1$\;
		$ \Omega_\nu^{(k_\nu)}  = \mr{SOMP}(\mQ(\vd_\nu^{(k_\nu)},0), \mY_{(r)},\mr{Iter} = 1)$ \;
		$\vd_\nu^{(k_\nu+1)}  = \mr{DICREF}(\vd_\nu^{(k_\nu)}, \Omega_\nu^{(k)}, \eta_\nu)$ \;
		}
	\While{halting criterion false}{
		$k_{\alpha}=k_{\alpha}+1$\;
		$\Omega_\alpha^{(k_{\alpha})}  = \mr{SOMP}\left(\mQ\left(\vd_\nu^{(k_\nu)}(\Omega_\nu^{(k_\nu)}), \vd_\alpha^{(k_\alpha)}\right), \mY_{(r)}\right)$ \;
		$\vd_\alpha^{(k_\alpha+1)}  = \mr{DICREF}(\vd_\alpha^{(k_\alpha)}, \Omega_\alpha^{(k)}, \eta_\alpha)$ 
		}
		$a_r = \exp(\vd_\alpha^{(k_\alpha)}(\Omega_\alpha^{(k_\alpha)}) + 2\pi\, \vd_\nu^{(k_\nu)}(\Omega_\nu^{(k_\nu)}))$
	}	
	\Return{$a_1, \ldots,a_R$}
\label{algo:STSM}
\end{algorithm}

\section{Multiple $R$-D Modes Estimation} \label{sec:multipleRD}

In the multiple tones case, sparse approximation algorithms yield suboptimal solutions when the coherence of the dictionary is high~\cite{Tropp:04}. This is a crucial point because the refinement procedure will increase the coherence with increasing $\ell$, which may prevent convergence even in the noiseless case. In the following, we present a low complexity algorithm that is accurate and robust in the presence of noise. The idea is to begin by an initialization step where $F$ single tone modal signals of order $R-1$ are extracted from the $R$-D signal. Then \fin an iterative technique is proposed \fin to improve this decomposition and estimate the parameters.

It is assumed that the frequencies are distinct in at least  one dimension with $M_r>F$. Then dimensions are permuted such that the dimension with distinct frequencies becomes the first one ($r=1$).  
\vspace{-3mm}
\subsection{\red From Multiple Tones to Multiple Single-Tone Signals  \fin}

According to (\ref{eq:Y-a-i-r}), $\ts{Y}$ can be written as
\begin{align}
\ts{Y} & = \ts{I}_{R+1,F} \argtimes{r=1}{R} \mv{A}_r \con_{R+1} \vc^\t \\   
 & = \ts{S}  \con_1 \mA_1 \label{eq:Y_Ys}
\end{align}
where $\ts{I}_{R+1,F}$ is the diagonal tensor of order $R+1$  and size $F\times F \times \cdots \times F$, containing ones on its diagonal,  and 
\begin{align}
 \ts{S}  = \ts{I}_{R+1,F} \argtimes{r=2}{R} \mv{A}_r \con_{R+1} \vc^\t    
\end{align}
is a complex tensor of order $R$ and size $F \times M_2 \times \cdots M_R$.  Similar expressions are evoked  in, among others, \cite{sun2012accurate}. The new tensor  $\ts{S}$  can also be written as the concatenation of $F$ tensors along the first dimension   
\begin{align}\label{eq:Ys_conc}
 \ts{S}  =   {\ts{S}}_{1} \sqcup_1 {\ts{S}}_{2} \sqcup_1 \cdots \sqcup_1 {\ts{S}}_{F}
\end{align}
where each  ${\ts{S}}_{f},  f=1,\ldots,F$ is a modal $(R-1)$-D signal of size $1\times M_2 \times \cdots \times M_R$ containing a single  $(R-1)$-D tone: 
\begin{equation}
 \ts{S}_{f} =  c_f~\mv{a}_{f,2} \out \mv{a}_{f,3} \out \cdots \out\mv{a}_{f,R}. 
\end{equation}

The singular value decomposition (SVD) of $\widetilde\mY_{(1)}$ yields 
\begin{align}
	\widetilde\mY_{(1)} = \mU \mSigm \mV^\h
	\label{eq_svd_Y1}
\end{align}
where matrices $\mU$ and $\mV$ contain respectively the left and right singular vectors of $\widetilde\mY_{(1)}$, and $\mSigm$ is a diagonal matrix containing the  singular values $\sigma_i, i = 1, \ldots, \min\{M_1, M_1' \}$ sorted in a decreasing order.  As the number of components in $\ts{Y}$ is equal to $F$, then an approximation of $\mY_{(1)}$, denoted  by $\hat\mY_{(1)}$, can be obtained using the first $F$ principal components of the SVD:
\begin{align}\label{eq:svdY1Truncated}
	\hat\mY_{(1)} = \mU_{F} \mSigm_{F} \mV_{F}^\h
\end{align}
where $\mU_{F}$ (resp. $\mV_{F}$) stands for the matrix formed with the first $F$ columns of $\mU$ (resp. $\mV$) and $\mSigm_F = \diag(\sigma_1, \ldots, \sigma_F)$. It can be established from~(\ref{eq:unfoldingYtilde}) and (\ref{eq:svdY1Truncated}) that $\mA_1$ and $\mU_F$ span the same subspace, and thus there exists an unknown nonsingular matrix $\mT$ that satisfies
\begin{equation}\label{eq:A1}
	\mA_1 = \mU_{F} \mT.
\end{equation}
Denote by $\underline\mM$ (resp. $\overline\mM$) the matrix obtained from $\mM$ by deleting the first (resp. last) row. By harnessing the Vandermonde structure of $\mA_1$, there exists a diagonal matrix $\mD$ such that $\underline\mA_1 = \overline\mA_1\mD$. Since  $\underline\mA_1 =  \underline\mU_F \mT$ and $\overline\mA_1 =  \overline\mU_F \mT$, then $\underline\mU_F \mT = \overline\mU_F\mT \mD$, which proves that matrix $\mT$ can be estimated by the eigenvectors of $\underline\mU_{F}^\dag\overline{\mU}_{F}$.

Thereby  $\ts{S}$  can be estimated from the noisy data and $\hat\mA_1$ using equation~(\ref{eq:Y_Ys}) as follows 
\begin{align}
 {\hat{\ts{S}}}  = \widetilde{\ts{Y}} \con_1 \hat\mA_1^\dag \label{eq:Yshat},
\end{align}
 then $\hat{\ts{S}}_{f}, f=1,\ldots,F$ are extracted from  ${\hat{\ts{S}}}$  according to~(\ref{eq:Ys_conc}). \soul Each $\ts{Y}_f = c_f\va_{f,1}\out \cdots \out\va_{f,R}$ can be estimated by $\bar{\ts{Y}}_f^{(0)} =  \hat{\ts{S}}_{f} \con_1 \hat{\va}_{f,1}$. \fin The sparse multigrid algorithm for single tone (STSM) can be applied on each  \soul $\bar{\ts{Y}}_{f}^{(0)}$, $f = 1,\ldots, F$ \fin to estimate the parameters of modes.  \red However, we propose in the following to improve the separated components using an iterative technique. \fin

\vspace{-2mm}
\subsection{\red  Improving the Estimation Accuracy \fin}

\red It is clear from (\ref{eq:Yshat}) that, in the noisy case, the error in estimating $\ts{S}$ (due to the estimation of $\mathbf{A}_1$) will propagate when estimating the parameters $a_{f,2}, \ldots, a_{f,R}$. Hence, we propose to improve iteratively the mode estimates. The following \soul procedure is executed to update estimates at each iteration $i=0,\ldots,K$
\begin{enumerate}
\item apply STSM to estimate $\va_{f,2},\dots,\va_{f,R}, f = 1,\ldots,F$ 
\begin{align}
\{\hat{a}_{f,2}, \ldots, \hat{a}_{f,R}\} & = \mr{STSM}(\bar{\ts{Y}}_{f}^{(i)} , \eta_\nu, \eta_\alpha, r= 2,\ldots,R) 
\label{eq_1proce_df1_afR}
\end{align}
\item estimate $c_f \va_{f,1}, f=1,\ldots,F$ by least squares using the already estimated $\va_{f,2},\dots,\va_{f,R}, f = 1,\ldots,F$
\begin{align}  \red \widehat{{c}_f {\va}_{f,1}} \soul  =  \bar{\mY}_{{f}_{(1)}}^{(i)} \left( (\hat{\va}_{f,R}\kron \cdots \kron \hat{\va}_{f,2})^\t \right)^\dag
\label{eq_2proce_df1_afR}
\end{align}
\item compute  $\hat{\ts{Y}}_{f}^{(i)} $ \begin{align} \hat{\ts{Y}}_{f}^{(i)}  =  \red \widehat{{c}_f {\va}_{f,1}} \soul  \out \hat{\va}_{f,2}\out \cdots \out \hat{\va}_{f,R} 
\label{eq_3proce_df1_afR}
\end{align}
\end{enumerate}
\red
where $\bar{\ts{Y}}_f^{(i)} =  \hat{\ts{Y}}_f^{(i-1)}  + \ts{R}_{f-1}^{(i)}$, $\ts{R}_{(f)}^{(i)} =  \ts{R}_{f-1}^{(i)} + \hat{\ts{Y}}_f^{(i-1)}  -   \hat{\ts{Y}}_f^{(i)}, f= 1,\ldots,F$, $\ts{R}_{0}^{(i)} \eqdef \ts{R}_{F}^{(i-1)}$, and $\ts{R}_{F}^{(0)} = \red \tilde{\ts{Y}} \soul - \sum_{f=1}^F \hat{\ts{Y}}_{f}^{(0)}$.  \soul This iterative scheme will be analyzed in the next section. \fin
\fin

\red Finally, \fin the algorithm we propose  (\emph{MTSM: Multiple Tones Sparse Method})  is summarized in Algorithm~\ref{algo:multipleModes}. Note that no association step of $R$-D modes is required. The initialization step consists in  initializing: i) $\hat\mA_1$ and  $\hat{\ts{S}}$  using (\ref{eq:A1}), (\ref{eq:Yshat}) and (\ref{eq:Ys_conc}), ii)  the estimated single tones \soul $\bar{\ts{Y}}_{f}^{(0)}, f=1,\ldots,F$. \fin 
 Note that the columns of $\hat{\mA}_1$ are iteratively updated without extracting the related modes, whereas the modes of the other dimensions are extracted at each iteration using (\ref{eq_1proce_df1_afR}).  Solely after the last iteration ($i=K$), the parameters of the first dimension are extracted using STSM algorithm. 
 $K$ denotes the maximum number of iterations, which is fixed to 2 in the simulations since no improvement was observed for $K>2$. 

\IncMargin{0mm}\begin{algorithm}[t]
\caption{Sparse multigrid method for multiple tones estimation  (MTSM) }
\footnotesize
\DontPrintSemicolon
\SetKwData{Left}{left}\SetKwData{This}{this}\SetKwData{Up}{up} \SetKwFunction{Union}{Union}
\SetKwFunction{FindCompress}{FindCompress} \SetKwInOut{Input}{input}\SetKwInOut{Output}{output}
\Input {A tensor $\widetilde{\ts{Y}} \in \mathbb{C}^{M_1\times \cdots \times M_R}$, $(\eta_\nu, \eta_\alpha)\in \mathbb{N}\times\mathbb{N}$}
\Output{Parameters of the multiple  $R$-D  modes : $\{a_{f,r}\}_{f=1,r=1}^{F,R}$}
{\bf initialization}: 
\begin{enumerate}
\item Compute $\hat\mA_1$ and  $\hat{\ts{S}}_{f}$, $f=1,\ldots,F$ using (\ref{eq:A1}),  (\ref{eq:Yshat}) and (\ref{eq:Ys_conc})
\item \soul $\bar{\ts{Y}}_{f}^{(0)} =  \hat{\ts{S}}_{f}  \con_1 \hat\va_{f,1}, $ $f=1,\ldots,F$  \fin
\end{enumerate}
\BlankLine
 \soul  For $f=1,\ldots,F$, compute $\hat{\ts{Y}}_{f}^{(0)}$ using ~(\ref{eq_1proce_df1_afR}), (\ref{eq_2proce_df1_afR}) and (\ref{eq_3proce_df1_afR}) \;
		
$\ts{R}_{F}^{(0)} \eqdef \ts{R}_{0}^{(1)} =  \tilde{\ts{Y}}  - \sum_{f=1}^{F}{\hat{\ts{Y}}_f^{(0)}}$\; \fin
		
\For{$i=1:K $}{
\For{$f=1:F $}{ \soul
$\bar{\ts{Y}}_{f}^{(i)} = \hat{\ts{Y}}_{f}^{(i-1)} +  \ts{R}_{f-1}^{(i)}$\;
 compute $\hat{\ts{Y}}_{f}^{(i)}$ using ~(\ref{eq_1proce_df1_afR}), (\ref{eq_2proce_df1_afR}) and (\ref{eq_3proce_df1_afR})\;
		$\ts{R}_{f}^{(i)} = \bar{\ts{Y}}_{f}^{(i)} - \hat{\ts{Y}}_{f}^{(i)}$,   if $f=F$, then  $ \ts{R}_{0}^{(i+1)}  \eqdef \ts{R}_{F}^{(i)} $ \;\fin
	}
	
	}
	For $f=1,\ldots,F$, extract $a_{f,1}$ using \;
		\quad $a_{f,1} = \mr{STSM}(\hat{\ts{Y}}_{f}^{(K)}+\ts{R}_F^{(K)}, \eta_\nu, \eta_\alpha, r=1)$
		
	\Return{$\{\hat{a}_{f,r}\}_{f=1,r=1}^{F,R}$}
	\label{algo:multipleModes}
\end{algorithm} 
\subsection{\red  Analysis of the Algorithm \fin}
\soul Following the separation step described in (\ref{eq_svd_Y1})--(\ref{eq:Yshat}), we can state that the algorithm \red yields the expected solution \soul when the SNR is sufficiently high. We want now to prove that the second \red stage (next iterations), \soul in addition to estimating the parameters from the single tones, 
 is also improving the estimation accuracy. \soul The general idea is inspired from greedy forward/backward sparse approximation, where the solution is refined by adding/removing atoms to/from the set of activated atoms. \soul The improvement of the estimates is stated by the \red following theorem. \soul

\begin{theorem}\label{proposition_algoMTSM} 
Assuming that the noise $\ts{E}$ is sufficiently small such that the ordering of the singular values in $\mSigm$ in (\ref{eq_svd_Y1}) is the same as the ordering of the corresponding singular values when $\ts{E}=0$. 
Using the procedure expressed by~(\ref{eq_1proce_df1_afR}), (\ref{eq_2proce_df1_afR}) and (\ref{eq_3proce_df1_afR}) to  estimate ${\ts{Y}}_{f}$ at iteration $i=0,\ldots,K$ 
\begin{align}
 \hat{\ts{Y}}_f^{(i)}  = \argmin{\ts{X}\in \mathcal{H}}{ \|\bar{\ts{Y}}_f^{(i)}    - \ts{X}} \|
\end{align}
\red where $\mathcal{H} = \{\ts{X} \in\mathbb{C}^{M_1\times\cdots\times M_R} | \ts{X} = \vb_1 \out \vb_2 \out \cdots \out \vb_R,$ $ \mathbf{b}_r \in \mathcal{P} \text{ for } r\neq 1\}$ with $\mathcal{P} = \{ \vv \in \mathbb{C}^{M_r} | \vv = [1, v, \ldots, v^{M_r-1}]^\t,  v = \exp(\beta+j \omega), \beta \in\mathbb{R}^-, \omega \in [0,2\pi)  \} $. \soul 
 Then, at each iteration $i, i=1,\ldots,K$ the residual is decreased: \red
\begin{align}
 \left\|  \tilde{\ts{Y}} - \hat{\ts{Y}}^{(i)} \right\| \leq \left\| \tilde{\ts{Y}} -  \hat{\ts{Y}}^{(i-1)} \right\| 
\label{eq_RFim1_RFi}
\end{align} 
where $ \hat{\ts{Y}}^{(i)} = \sum_{f=1}^F \hat{\ts{Y}}_f^{(i)}$. \soul
\begin{proof}
See Appendix~\ref{appendix_proof_proposition_algoMTSM}. \hfill $\blacksquare$
\end{proof}
\soul
Figure~\ref{fig_MTSM_K2_K0} depicts a comparison between results obtained by the MTSM algorithm with two different values of $K\in\{0,2\}$. The results show that MTSM with the improving step yields accurate estimates as compared to MTSM without the improving step.
\begin{figure}[t]
\centerline{\includegraphics[width=.9\linewidth]{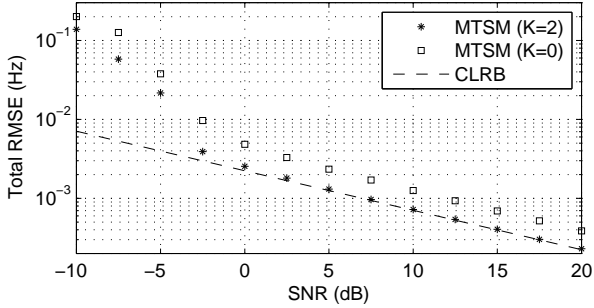}}
\caption{Frequency total root-mean square error for a 3-D signal containing 3 modes with identical modes in two dimensions and close modes in the first dimension (Signal~\#4 in Table~\ref{tab:tab2}). $N(0) = 20, \eta_\nu = 21, \eta_\beta = 11, (M_1,M_2,M_3) = (10,10,10)$. 100 Monte-Carlo.}
\label{fig_MTSM_K2_K0} \vspace{-5mm}
\end{figure}

\end{theorem}

   \fin

\vspace{-3mm}
\subsection{\soul Identifiability \fin}
 \soul Based on the assumptions under which Algorithm~\ref{algo:multipleModes} is operating, the identifiability condition can be stated as $F < M_1$ and $\min\{M_2,....,M_R\} \geq 2$.
 \soul In \cite{jiang2001almost}, the condition is $M_r \geq 4, r=1,\ldots, R$, and $F \leq \left\lfloor \frac{M_1}{2}\right\rfloor\prod_{r=1}^R\left\lceil\frac{M_r}{2}\right\rceil$. 

We note that, when $M_r \geq 4, r=1,\ldots, R$, the number of identifiable modes is slightly smaller than in \cite{jiang2001almost}, but the proposed algorithm is able to outperform the conventional methods in terms of computational complexity and accuracy.  In addition, another advantage of the proposed algorithm is clear when the number of samples in one or more dimensions is less than 4 (i.e. $M_r < 4$), where identifiability in~\cite{jiang2001almost} is not satisfied. This latter case (i.e. $\exists r, M_r < 4$) can be encountered in signal processing applications when the size of one or more diversities (dimensions in our formulation problem) is less than~4. \fin

\vspace{-0mm}

\section{Cram\'er-Rao Lower Bounds for $R$-D Cisoids in Noise} \label{sec:crlb}

In this section, we derive the expressions of the CRLB for the parameters of $R$-D damped exponentials in Gaussian white noise. We then give the CRLB in the cases of single damped and undamped $R$-D cisoids. We consider the $R$-D sinusoidal model given in (\ref{eq:rd_signal_model}). Let
\begin{align*}
\boldsymbol{\theta}=[\omega_{1,1}\ldots\omega_{1,R} & \quad \omega_{2,1}\ldots\omega_{F,R}\quad \alpha_{1,1}\ldots\alpha_{1,R}\\ &\quad \alpha_{2,1}\ldots\alpha_{F,R} \quad \lambda_1\ldots\lambda_F\quad \phi_1\ldots\phi_F]^\t
\end{align*}
be the unknown parameter vector. The aim here is to derive the CRLB of the parameters in $\boldsymbol{\theta}$.

The joint probability density function (pdf) of $\tilde{\mv{y}}$ is
\begin{equation}
p(\tilde{\mv{y}};\boldsymbol{\theta}) = \frac{1}{(\sigma^2\pi)^M}\exp\left\{-\frac{1}{\sigma^2}(\tilde{\mv{y}}-\boldsymbol{\mu}(\boldsymbol{\theta}))^\h(\tilde{\mv{y}}-\boldsymbol{\mu}(\boldsymbol{\theta}))\right\} \label{eq:jointpdf}
\end{equation}
where $\boldsymbol{\mu}(\boldsymbol{\theta})$ is the noise-free part of $\mathbf{y}$ and
\begin{align}
\tilde{\mv{y}}=&[\tilde{y}(1,\ldots,1,1),\ldots, \tilde{y}(1,\ldots,1,M_R),\nonumber\\
 &\tilde{y}(1,\ldots,2,1),\ldots, \tilde{y}(1,\ldots,2,M_R),\nonumber\\
 &\ldots, \tilde{y}(M_1,\ldots,M_R)]^\t.
\end{align}
The $i$th entry of $\boldsymbol{\mu}(\boldsymbol{\theta})$ can be written as:
\begin{equation}
\boldsymbol{\mu}(\boldsymbol{\theta})_i = \sum_{f=1}^F c_f \prod_{r=1}^R a_{f,r}^{t_{i,r}},
\end{equation}
for $i=1,\ldots,M$, where
\begin{equation}
t_{i,r} = \left\lfloor \frac{i-1}{\prod_{\ell=r+1}^{R}M_\ell} \right\rfloor \quad \mod \quad M_r,
\end{equation}
and $\lfloor\cdot\rfloor$ is the floor function. In the following, we derive the expressions of the CRLB in the general case ($F>1$) and then we deduce the result corresponding to a single $R$-D modal signal ($F=1$).

\subsection{Derivation of the CRLB}

Given the joint pdf in (\ref{eq:jointpdf}), the $(k,l)$ entry of the Fisher information matrix is~\cite{Yao:95,Kay:93}:
\begin{equation}
[\mathbf{F}(\boldsymbol{\theta})]_{kl} = \frac{2}{\sigma^2} \mathrm{Re} \left\{ \left[ \frac{\partial\boldsymbol{\mu}(\boldsymbol{\theta})}{\partial\theta_k} \right]^\h \frac{\partial\boldsymbol{\mu}(\boldsymbol{\theta})}{\partial\theta_l}\right\}.
\end{equation}
We now express the derivatives $\partial\boldsymbol{\mu}(\boldsymbol{\theta})_i/\partial\theta_k$ for $i=1,\ldots,M$ and $k=1,\ldots,2RF+2F$.
\begin{itemize}
\item For $k=1,\ldots,RF$, we have
\begin{equation}
\frac{\partial\boldsymbol{\mu}(\boldsymbol{\theta})_i}{\partial\theta_k} = jt_{i,r(k)} c_{f(k)} \prod_{r=1}^R a_{f(k),r}^{t_{i,r}}
\end{equation}
with $r(k)=[(k-1)\,\mod\,R]+1$ and $f(k)=\lfloor(k-1)/R\rfloor+1$.
\item For $k=RF+1,\ldots,2RF$:
\begin{equation}
\frac{\partial\boldsymbol{\mu}(\boldsymbol{\theta})_i}{\partial\theta_k} = t_{i,r(k)} c_{f(k)} \prod_{r=1}^R a_{f(k),r}^{t_{i,r}}
\end{equation}
with $r(k)=[(k-RF-1)\,\mod\,R]+1$ and $f(k)=\lfloor(k-RF-1)/R\rfloor+1$.
\item For $k=2RF+1,\ldots,2RF+F$:
\begin{equation}
\frac{\partial\boldsymbol{\mu}(\boldsymbol{\theta})_i}{\partial\theta_k} = e^{j\phi_{f(k)}} \prod_{r=1}^R a_{f(k),r}^{t_{i,r}}
\end{equation}
where $f(k)=k-2RF$.
\item For $k=2RF+F+1,\ldots,2RF+2F$:
\begin{equation}
\frac{\partial\boldsymbol{\mu}(\boldsymbol{\theta})_i}{\partial\theta_k} = jc_{f(k)} \prod_{r=1}^R a_{f(k),r}^{t_{i,r}}
\end{equation}
where $f(k)=k-2RF-F$.
\end{itemize}
Hence, the $M\times (2RF+2F)$ matrix $\partial\boldsymbol{\mu}(\boldsymbol{\theta})/\partial\boldsymbol{\theta}$ expresses as
\begin{align}
\frac{\partial\boldsymbol{\mu}(\boldsymbol{\theta})}{\partial\boldsymbol{\theta}} 
 &= \underbrace{[j\mathbf{Z}'\boldsymbol{\Phi}\quad \mathbf{Z}'\boldsymbol{\Phi}\quad \mathbf{Z}\boldsymbol{\phi}\quad j\mathbf{Z}\boldsymbol{\phi}]}_{\mathbf{V}}\cdot \underbrace{\textrm{blkdiag}(\boldsymbol{\Lambda},\boldsymbol{\Lambda},\mathbf{I}_F,\boldsymbol{\lambda})}_{\mathbf{S}}
\end{align}
where
\begin{align}
\mathbf{Z}' & = [\mathbf{Z}'_1,\ldots,\mathbf{Z}'_F]\in\mathbb{C}^{M\times RF},\text{with } \mathbf{Z}'_f(i,l)=t_{i,l}\prod_{r=1}^{R} a_{f,r}^{t_{i,r}},\label{eq:Z'}\\
\boldsymbol{\Lambda} &= \mathrm{blkdiag}(\lambda_1\mathbf{I}_R,\ldots,\lambda_F\mathbf{I}_R)\in\mathbb{R}^{RF\times RF},\\
\boldsymbol{\Phi} &= \mathrm{blkdiag}(e^{j\phi_1}\mathbf{I}_R,\ldots,e^{j\phi_F}\mathbf{I}_R)\in\mathbb{C}^{RF\times RF},\\
\mathbf{Z} & = [\mathbf{z}_1,\ldots,\mathbf{z}_F]\in\mathbb{C}^{M\times F},\text{with } \mathbf{z}_f(i)=\prod_{r=1}^{R}a_{f,r}^{t_{i,r}},\label{eq:Z}\\
\boldsymbol{\lambda} &= \mathrm{diag}([\lambda_1,\ldots,\lambda_F])\in\mathbb{R}^{F\times F},\\
\boldsymbol{\phi} &= \mathrm{diag}([e^{j\phi_1},\ldots,e^{j\phi_F}])\in\mathbb{C}^{F\times F}.
\end{align}
Finally, the inverse of the Fisher information matrix is
\begin{equation}
\mathbf{F}^{-1}(\boldsymbol{\theta}) = \frac{\sigma^2}{2}\mathbf{S}^{-1} \left[\mathrm{Re}\{\mathbf{V}^\h\mathbf{V}\}\right]^{-1} \mathbf{S}^{-1} = \frac{\sigma^2}{2}\mathbf{S}^{-1} \mathbf{W} \mathbf{S}^{-1}
\end{equation}
where $\textrm{Re}\{\cdot\}$ stands for the real part. The CRLB of $\theta_k$ is given by $[\mathbf{F}^{-1}(\boldsymbol{\theta})]_{kk}$. More explicitly, for $f=1,\ldots,F$ and $r=1,\ldots,R$:
\begin{align}
\mathrm{CRLB}(\omega_{f,r}) &= \frac{2\sigma^2\mathbf{W}_{R(f-1)+r,R(f-1)+r}}{\lambda_f^2}\\
\mathrm{CRLB}(\alpha_{f,r}) &= \frac{2\sigma^2\mathbf{W}_{RF+R(f-1)+r,RF+R(f-1)+r}}{\lambda_f^2}\\
\mathrm{CRLB}(\lambda_{f}) &= 2\sigma^2\mathbf{W}_{2RF+f,2RF+f}\\
\mathrm{CRLB}(\phi_{f}) &= \frac{2\sigma^2\mathbf{W}_{2RF+F+f,2RF+F+f}}{\lambda_f^2}
\end{align}

\begin{theorem}
For the general $R$-D exponential process, the CRLB's for $f=1,\ldots,F$ and $r=1,\ldots,R$ satisfy
\begin{align}
\mathrm{CRLB}(\omega_{f,r}) &= \mathrm{CRLB}(\alpha_{f,r})\\
\mathrm{CRLB}(\lambda_{f}) &= \lambda^2 \mathrm{CRLB}(\phi_{f})
\end{align}
\end{theorem}
\begin{proof}
It is based on the special block structure of matrix $\mathrm{Re}\{\mathbf{V}^\h\mathbf{V}\}$ (see for instance~\cite{Yao:95}).
\end{proof}

\subsection{Single Mode Case}
In this section, the CRLB's will be simplified in the case of a single $R$-D modal signal ($F=1$) to obtain more precise details on their parameter dependency. For the sake of simplicity, the subscripts denoting the mode $f=1$ will be omitted. First, assume that $|a_r|=\exp(\alpha_r)<1$. We shall express the products $\mathbf{Z}'^\h\mathbf{Z}'$, $\mathbf{Z}^\h\mathbf{Z}$ and $\mathbf{Z}'^\h\mathbf{Z}$. After some calculations, we get:
\begin{align} 
[\mathbf{Z}'^\h\mathbf{Z}']_{nk} 
 &= \prod_{\substack{r=1 \\ r\neq n,k}}^{R}\left(\frac{1-|a_r|^{2M_r}}{1-|a_r|^2}\right)\nonumber\\
 &\times \begin{cases}
 \displaystyle \sum_{m=0}^{M_n-1}m|a_n|^{2m} \displaystyle \sum_{m=0}^{M_k-1}m|a_k|^{2m}, & \text{if } n\neq k\\
 \displaystyle \sum_{m=0}^{M_n-1}m^2|a_n|^{2m}, & \text{if } n=k
 \end{cases}\\
\mathbf{Z}^\h\mathbf{Z} 
 &= \prod_{r=1}^{R}\left(\frac{1-|a_r|^{2M_r}}{1-|a_r|^2}\right)\\
[\mathbf{Z}'^\h\mathbf{Z}]_n 
 &= \prod_{\substack{r=1 \\ r\neq n}}^{R}\left(\frac{1-|a_r|^{2M_r}}{1-|a_r|^2}\right)\times \sum_{m=0}^{M_n-1}m|a_n|^{2m}.
\end{align}
Denoting $M^{(\alpha)}=\prod_{r=1}^{R}(1-|a_r|^{2M_r})/(1-|a_r|^2)$, $q_1(n)=\sum_{m=0}^{M_n-1}m|a_n|^{2m}/\sum_{m=0}^{M_n-1}|a_n|^{2m}$ and $q_2(n)=\sum_{m=0}^{M_n-1}m^2|a_n|^{2m}/\sum_{m=0}^{M_n-1}|a_n|^{2m}$, we then obtain:
\begin{align}
[\mathbf{P}]_{nk} &= M^{(\alpha)}\times \begin{cases}
 q_1(n)q_1(k), &\text{if } n\neq k\\
 q_2(n), & \text{if } n=k
 \end{cases}\\
\mathbf{G} &= M^{(\alpha)}\\
[\mathbf{Q}]_n &= M^{(\alpha)} q_1(n),
\end{align}
and
\begin{equation}
\mathrm{Re}\{\mathbf{V}^\h\mathbf{V}\} = \begin{bmatrix}
\mathbf{P} & 0 & 0 & \mathbf{Q}\\
0 & \mathbf{P} & \mathbf{Q} & 0\\
0 & \mathbf{Q}^\t & \mathbf{G} & 0\\
\mathbf{Q}^\t & 0 & 0 & \mathbf{G}
\end{bmatrix}.
\end{equation}
The inversion of $\mathrm{Re}\{\mathbf{V}^\h\mathbf{V}\}$ yields the following expressions of the CRLB's:
\begin{multline}
\mathrm{CRLB}(\omega_{r}) =\mathrm{CRLB}(\alpha_{r}) =\frac{\sigma^2}{2\lambda^2 M^{(\alpha)}}  \\ 
 \quad \times \frac{(1-|a_r|^2)^2(1-|a_r|^{2M_r})^2}{[-M_r^2|a_r|^{2M_r}(1-|a_r|^2)^2+|a_r|^2(1-|a_r|^{2M_r})^2]},
\end{multline}
\begin{multline}
\frac{\mathrm{CRLB}(\lambda)}{\lambda^2} =\mathrm{CRLB}(\phi) =\frac{\sigma^2}{2\lambda^2 M^{(\alpha)}} \\
\quad \times \left(1+\sum_{r=1}^R \frac{q_1^2(r)}{q_2(r)-q_1^2(r)}\right).
\end{multline}
Finally, for a single $R$-D purely harmonic signal ($\alpha_r=0,\forall r$), we have $M^{(\alpha)}=\prod_{r=1}^R M_r=M$ and taking the limit of the CRLB's when $\alpha_r\to0$ leads to:
\begin{align}
\lim_{\alpha_r\to 0} \mathrm{CRLB}(\omega_{r}) &= \frac{6\sigma^2}{\lambda^2 M (M_r^2-1)}\label{eq:crlb-w1}\\
\lim_{\alpha_r\to 0} \frac{\mathrm{CRLB}(\lambda)}{\lambda^2} &= \frac{\sigma^2}{2\lambda^2 M}\left(1+3\sum_{r=1}^R \frac{M_r-1}{M_r+1}\right).
\end{align}
Hence, for the undamped case, our result in (\ref{eq:crlb-w1}) is consistent with~\cite{sun2012accurate}.

\section{Simulation Results} \label{sec:simul}


Numerical simulations have been carried out to assess the performances of the proposed method for 2-D and 3-D modal signals in the presence of white Gaussian noise. The performances are measured by the total root-mean square error (RMSE) on estimated parameters and the computational time. The total RMSE is defined as
$\mathrm{RMSE}_{\mathrm{total}}=\sqrt{\frac{1}{RF} \mathbb{E}_p\left\{ \sum_{r=1}^R\sum_{f=1}^F(\xi_{f,r}-\hat{\xi}_{f,r})^2\right\} }$
where $\hat{\xi}_{f,r}$ is an estimate  of $\xi_{f,r}$, and $\mathbb{E}_p$ is the average on $p$ Monte-Carlo trials. In our simulations, $\xi_{f,r}$ can be either a frequency or a damping factor.

\subsection{RMSE for 2-D and 3-D Signals}

\begin{experiment} \label{exprim:diff_settings}
to show the interest of the multigrid scheme, this experiment presents the results obtained on Signal \#1 with different multigrid levels and different initial grids. Signal \#1 is a single tone 2-D modal signal of size $10 \times 10$ whose parameters are presented in Table~\ref{table:singleTone_2d}. The number of multigrid levels is fixed to $L=2$, i.e., $\ell = 0,1,2$. Then the results are presented as a function of the number of atoms in the initial dictionaries \soul $N(0)$ \fin and the number of atoms \soul $\eta_\nu$ or $\eta_\alpha$ \fin added at each level $\ell$. 
The results we obtain for the first step, i.e., for the harmonic estimation, are presented in Figure~\ref{fig:N0_eta_nu}. We can observe that the frequency RMSE obtained with the $R$-D sparse algorithm can reach the CRLB using a uniform initial harmonic dictionary of 10 atoms and  $\soul \eta_{\nu} \fin = 31$  (Figure~\ref{fig:N0_eta_nu}.a).  Figure~\ref{fig:N0_eta_nu}.b shows that the frequency RMSE is improved at low SNR if the initial dictionary contains  20  atoms, and reaches the optimal estimates with  $\soul \eta_{\nu} \fin = 21$. 
Figure~\ref{fig:eta_beta} shows the damping factor RMSE obtained by $R$-D sparse algorithm using different settings of the initial damping factor dictionary and \soul$\eta_{\alpha}$\fin.  We can observe that  the damping factor RMSE depends on the number of atoms in the dictionary, the more atoms the better. At low SNR, the RMSE also depends $\beta_{\min}$. Therefore, it is better to choose $\beta_{\min}$ with small absolute value if we have a prior knowledge of the interval of damping factors in the signal. \soul In general, the estimation error is of order $ \frac{1}{N(0) \eta^2}$. For instance, in the frequency step estimation, we recommend to chose $N(0) $ to be greater than or equal  to $\frac{3}{2}M_r$ if we want a good accuracy at lower SNR levels. Otherwise, we can set $N(0)  = M_r$. Once $N(0) $ is set, $\eta$ can be chosen with respect to the desired accuracy. Let $\varepsilon$ be the desired estimation error,  then $\varepsilon = \frac{1}{N(0)  \eta^2}$ and we can set $\eta = \frac{1}{\sqrt{\varepsilon N(0) }}$.   \fin
\end{experiment}
\begin{table} [t]
\caption{2-D parameters of Signal \#1}
\centering
\begin{tabular}{|c||c|c||c|c||c|} 
\hline 
$f$ & $\nu_{f,1}$ & $\alpha_{f,1}$ & $\nu_{f,2}$ & $\alpha_{f,2}$ & $c_f$\\\hline \hline
1 & $0.22$ & $-0.011$ & $0.34$ & $-0.015$ & 1 \\ \hline
\end{tabular}
\label{table:singleTone_2d}
\end{table}

\begin{figure}[t]
\psfrag{N0 = 10, eta = 5}[cc][cc]{\tiny \hspace{-2mm } $N(0) = 10, \eta_\nu = 5$} 
\psfrag{N0 = 10, eta = 11}[cc][cc]{\tiny \hspace{-2mm } $N(0) = 10, \eta_\nu = 11$} 
\psfrag{N0 = 10, eta = 31}[cc][cc]{\tiny \hspace{-2mm } $N(0) = 10, \eta_\nu = 31$} 
\psfrag{N0 = 20, eta = 5}[cc][cc]{\tiny  \hspace{-2mm } $N(0) = 20, \eta_\nu = 5$} 
\psfrag{N0 = 20, eta = 11}[cc][cc]{\tiny \hspace{-2mm } $N(0) = 20, \eta_\nu = 11$} 
\psfrag{N0 = 20, eta = 21}[cc][cc]{\tiny \hspace{-2mm } $N(0) = 20, \eta_\nu = 21$} 
\psfrag{CRLB}[cc][cc]{\tiny CRLB}  
\subfigure[$N(0) = 10$]{\centerline{\epsfig{figure=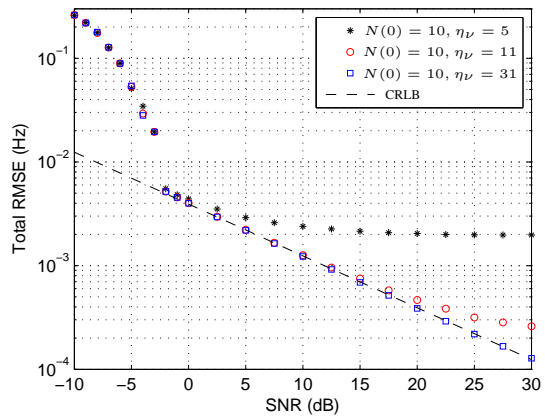, width=1.1\largeur}}} \vspace{-3mm}
\subfigure[$N(0) = 20$]{\centerline{\epsfig{figure=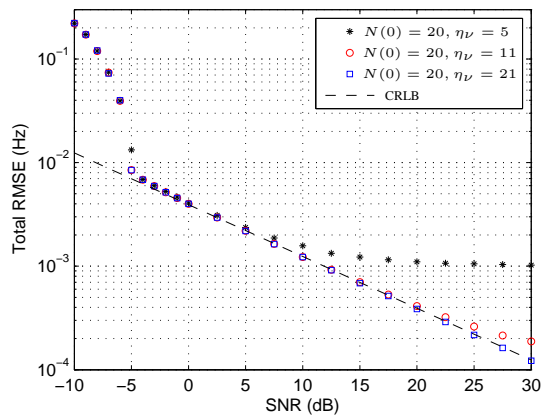, width=1.1\largeur}}}
\caption{Frequency RMSE using $R$-D sparse algorithm with different $\eta_{\nu}$.  2-D signal containing a single tone (Signal \#1). $(M_1,M_2) = (10,10).$ 1000 Monte-Carlo trials. (a) The initial harmonic dictionary contains $N(0) = 10$ atoms, (b) $N(0) = 20$ atoms. }
\label{fig:N0_eta_nu} 
\end{figure}

\begin{figure}[t]
\psfrag{N0 = 04, eta= 05, betaMin = -2}[cc][cc]{\tiny $N(0) = 04,   \eta_{\alpha} = 05, \beta_{\min} = -2$ }
\psfrag{N0 = 20, eta= 21, betaMin = -2}[cc][cc]{\tiny $N(0) = 20,   \eta_{\alpha} = 21, \beta_{\min} = -2$ }
\psfrag{N0 = 10, eta= 11, betaMin = -0.05}[cc][cc]{\tiny $N(0) = 10,   \eta_{\alpha} = 11, \beta_{\min} = -0.05$ }
\psfrag{CRLB}[cc][cc]{\tiny  CRLB }
\centerline{\includegraphics[width=1.1\largeur]{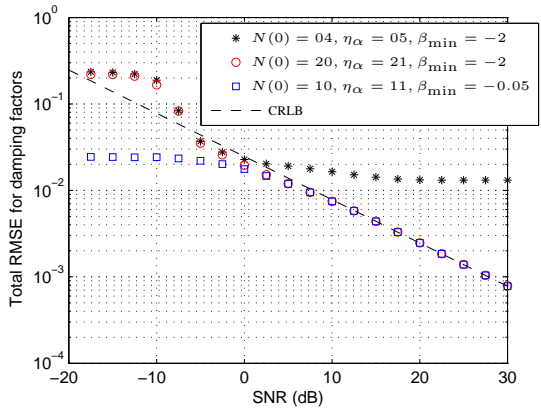}}
\caption{Damping factor RMSE using $R$-D sparse algorithm with different $\eta_{\alpha}$, $\beta_{\min}$ and $N(0)$. 2-D signal containing a single tone (Signal \#1). $(M_1,M_2) = (10,10).$ 1000 Monte-Carlo trials.}
\label{fig:eta_beta}
\end{figure}
 
In the rest of this section, the proposed algorithms are compared with 2-D ESPRIT~\cite{Rouquette2001}, Tensor-ESPRIT~\cite{haardt2008higher}, PUMA~\cite{so2010efficient} and TPUMA~\cite{sun2012accurate}. If the $R$-D signal contains one tone then Algorithm~\ref{algo:STSM}  (STSM)  is used, otherwise Algorithm~\ref{algo:multipleModes}  (MTSM)  is used. Thus, to facilitate notation, both proposed algorithms, Algorithm~\ref{algo:STSM} and Algorithm~\ref{algo:multipleModes}, will be called $R$-D sparse. 
For the proposed method, the initial grid used to build the harmonic dictionary is the same for all dimensions; it contains $50$ frequency points uniformly distributed over the interval $[0,1)$ and 10 damping factors $\beta\in [-0.05, 0]$.  To simulate a random dictionary, at each run, the frequency grid is perturbed by a small random quantity. As a consequence of experiment~\ref{exprim:diff_settings}, we use the following settings $ (L, \eta_\mu, \eta_\beta)= (2,21,11)$. The number of iterations in Algorithm~\ref{algo:multipleModes} is set to $K=2$ because no improvement was observed for $K>2$. 

  Since the proposed method is applied directly on data without using spatial smoothing, i.e., it does not require the construction of a large matrix or an augmented order tensor, then a relevant comparison will be with algorithms that do not use spatial smoothing. Thereby, in the next experiments, the proposed algorithm is compared to PUMA~\cite{so2010efficient} and TPUMA~\cite{sun2012accurate}, which are algorithms that do not require spatial smoothing. We also report comparisons with 2-D ESPRIT~\cite{Rouquette2001} and Tensor-ESPRIT~\cite{haardt2008higher}, which need spatial smoothing. 

\subsubsection{Single tone $R$-D modal signal}
\begin{experiment}
This experiment tends to show the efficiency of the proposed algorithm in estimating parameters of single tone $R$-D modal signals. We simulate a 2-D signal of size $10 \times 10$ (Signal \#1) whose parameters are presented in Table~\ref{table:singleTone_2d}. Our $R$-D sparse algorithm is compared to 2-D ESPRIT~\cite{Rouquette2001} and PUMA~\cite{so2010efficient}. For each level of noise, 1000 Monte-Carlo trials are performed. Figure~\ref{fig:singleTone} shows the obtained results. We can observe that: i) the proposed algorithm and PUMA reach the CRLB and outperform 2-D ESPRIT, ii) $R$-D sparse outperform PUMA in SNR less than 3~dB.     
\end{experiment}
\begin{figure}[t]
\centerline{\includegraphics[width=\largeur]{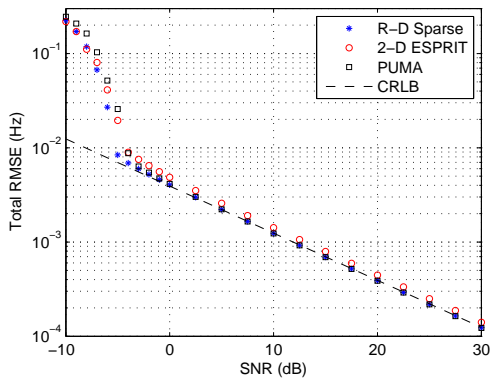}} 
\caption{Frequency total root-mean square error for a 2-D signal containing a single tone (Signal \#1). $(M_1,M_2) = (10,10).$ 1000 Monte-Carlo trials.}
\label{fig:singleTone} \vspace{-3mm}
\end{figure}

\subsubsection{Multiple tones R-D modal signals}
Several configurations are studied in the case of multiple tones to compare the proposed algorithm with Tensor-ESPRIT~\cite{haardt2008higher} and TPUMA~\cite{sun2012accurate}. These configurations \soul(Experiments \ref{experim:3d_twoSeprateModes}, \ref{experim:3d_3Modes_TwoIdentical}, \ref{experim:3modes_Idetical_LessThenFourier}
) \fin are summarized in Table~\ref{tab:experiments_resume}, in which the number of modes and the distance between frequencies in different dimensions are varied. $\Delta_{\mr{Fr}}$ denotes the Rayleigh frequency resolution limit, which has the same value in all dimensions because $M_1= M_2=M_3$. \soul In Experiment~\ref{experim:4modes_sizeLessThanF} we examine the case when the size of only one dimension is larger than 4, i.e., the identifiability condition of \cite{jiang2001almost} is not satisfied. The parameters of the used signals  are given in Table~\ref{tab:tab2}. \fin 
 
\begin{table}[t]
\caption{Different configurations for experiments 3 through 5} \label{tab:experiments_resume} \vspace{-2mm}
\centering
\begin{tabular}{|c||c|c|c|c|} 
\hline 
& $F$ & Dimension 1 & Dimension 2 & Dimension 3 \\ \hline \hline

Exp. \ref{experim:3d_twoSeprateModes} & 2 & $\Delta \nu > \Delta_{\mr{Fr}}$ & $\Delta \nu > \Delta_{\mr{Fr}}$ &$\Delta \nu > \Delta_{\mr{Fr}}$ \\ \hline

Exp. \ref{experim:3d_3Modes_TwoIdentical} & 3 & $\Delta \nu \geq \Delta_{\mr{Fr}}$ & $\exists$ identical modes & $\exists$ identical modes  \\ \hline
Exp. \ref{experim:3modes_Idetical_LessThenFourier} & 3 & $\Delta \nu < \Delta_{\mr{Fr}}$ & $\exists$ identical modes  &  $\exists$ identical modes  \\ \hline
\end{tabular}\vspace{-3mm}
\end{table}
\begin{table}[t]
\caption{3-D parameters of Signal \#2 through \#5} \label{tab:tab2}
\centering
\begin{tabular}{|c||l|l||l|l||l|l||l|} 
\hline 
Signal & $\nu_{f,1}$ & $\alpha_{f,1}$ & $\nu_{f,2}$& $\alpha_{f,2}$ & $\nu_{f,3}$ & $\alpha_{f,3}$ & $c_f$\\ \hline \hline
\#2 & $0.40$ & $-0.01$ & $0.1$ & $-0.01$ & $0.1$ & $-0.01$ & 1 \\ \cline{2-8}
 & $0.20$ & $-0.01$ & $0.3$ & $-0.15$ & $0.25$ & $-0.01$ & 1 \\ \hline\hline
\#3 & $0.30$ & $-0.01$ & $ 0.31 $ & $ -0.01 $ & $0.22$ & $-0.01$ & 1 \\ \cline{2-8}
 & $0.10$ & $-0.01$ & $0.45$ & $-0.015$ & $ 0.11 $ & $ -0.01 $ & 1 \\ \cline{2-8}
 & $0.20$ & $-0.01$ & $ 0.31 $ & $ -0.01 $ & $ 0.11 $ & $ -0.01 $ & 1 \\ \hline\hline
\#4 & $0.28$ & $-0.01$ & $ 0.31 $ & $ -0.01 $ & $0.22$ & $-0.01$ & 1 \\ \cline{2-8}
 & $0.12$ & $-0.01$ & $0.45$ & $-0.015$ & $ 0.11 $ & $ -0.01 $ & 1 \\ \cline{2-8}
 & $0.20$ & $-0.01$ & $ 0.31 $ & $ -0.01 $ & $ 0.11 $ & $ -0.01 $ & 1 \\ \hline\hline
\#5 & $0.30$ & $-0.01$ & $ 0.1 $ & $ -0.01 $ & $0.1$ & $-0.01$ & 1 \\ \cline{2-8}
 & $0.13$ & $-0.01$ & $0.45$ & $-0.015$ & $ 0.4 $ & $ -0.01 $ & 1 \\ \cline{2-8}
 & $0.20$ & $-0.01$ & $ 0.31 $ & $ -0.01 $ & $ 0.1 $ & $ -0.01 $ & 1 \\ \cline{2-8}
 & $0.42$ & $-0.012$ & $0.22$ & $-0.01 $ & $0.32$ & $ -0.01$ & 1 \\ \hline
\end{tabular}\vspace{-3mm}
\end{table}

\begin{experiment}
\label{experim:3d_twoSeprateModes}
In this experiment, we simulate a 3-D signal (Signal \#2) of size $8 \times 8 \times 8$ and containing two modes whose frequencies in each dimension are well separated. Parameters of the signal are given in Table~\ref{tab:tab2}. Figure~\ref{fig:two_3D_modes} shows the obtained results. Here, the proposed method performs as TPUMA. Tensor-ESPRIT yields slightly worse RMSE.  
\end{experiment}

\begin{experiment}\label{experim:3d_3Modes_TwoIdentical}
3-D signal of size $10 \times 10 \times 10$ containing three 3-D modes (Signal \#3). Note that there exists identical modes in two dimensions and frequencies in the first dimension are separated by $1/M_1$. Figure~\ref{fig:3modes_Identical_In_towDimdension} shows the results. In this configuration, TPUMA and Tensor-ESPRIT give similar results and the proposed method performs better for all SNR levels.
\end{experiment}

\begin{experiment}\label{experim:3modes_Idetical_LessThenFourier}
3-D signal of size $10 \times 10 \times 10$ containing three 3-D modes (Signal \#4). Note that there exists identical modes in two dimensions and frequencies in the first dimension are separated by less than $1/M_1$. The results are shown on Figure~\ref{fig:3modes_Identical_In_TwoDimdension_Close}. Here again the proposed $R$-D sparse approach performs better than TPUMA and Tensor-ESPRIT. Observe also that Tensor-ESPRIT outperforms TPUMA in this configuration (close frequencies and identical modes in dimensions 2-3).
\end{experiment}


\soul
\begin{experiment} \label{experim:4modes_sizeLessThanF}
Results on Signal \#5 of size $10\times3\times3$ containing 4 modes  are given in Figure~\ref{fig_signal_10x3x3}. We observe that the proposed method outperforms TPUMA algorithm mainly in low SNR levels. 
\end{experiment} 
\fin

\subsection{Numerical Complexity}
It is known that in the case of 1-D signals of size $M$, OMP costs $\mathrm{O}(NFM)$ in terms of multiplications~\cite{Tropp2010computational}; $F$ is the sparsity (number of components) and $N$ is the number of atoms in the dictionary. For a $M$-measurements $R$-D signal, the complexity of the STSM algorithm over a set of $L$ multigrid levels is $\mr{O}(MN L R)$, assuming that the number of dictionary atoms is maintained constant (equal to $N$) over all levels. Regarding the approach proposed in Algorithm~\ref{algo:multipleModes}, the main operations are the call of STSM and the update of \soul $\widehat{c_f{\va}_{f,1}}  =  \bar{\mY}_{{f}_{(1)}}^{(i)} \left( (\hat{\va}_{f,R}\kron \cdots \kron \hat{\va}_{f,2})^\t \right)^\dag$ 
\fin which has a complexity of  $\mr{O}(M )$ since \soul $\left( (\hat{\va}_{f,R}\kron \cdots \kron \hat{\va}_{f,2})^\t \right)^\dag$ \fin is a row of length $\prod_{r=2}^R M_r$ and \soul $\bar{\mY}_{{f}_{(1)}}^{(i)}$ \fin is a matrix of size $M_1\times \prod_{r=2}^R M_r$. Therefore, the whole complexity of the proposed algorithm is $\mr{O}\left((  N L (F(R-1)K+ 1)+ F K ) M \right)$, which is linear in the number of measurements $M$. The complexity of the Tensor-ESPRIT algorithm with spatial smoothing is mainly related to that of the SVD which is at least $\mathrm{O}(k_tF(R+1)PM)$ where $k_t$ is a constant depending on the implementation of the SVD algorithm. Here $P=\prod_{r=1}^R P_r$ where $\{P_r\}_{r=1}^R$ are design parameters used to get smoothed measurements (see~\cite{haardt2008higher}). The accuracy of the estimates provided by ESPRIT depends on these parameters. Since the optimal value for $P_r$ is a fraction of $M_r$ (\emph{e.g.}~\cite{Lemma:03,Djermoune:2009perturbation,Hua1990}), the complexity of the SVD step is, in fact, $\mathrm{O}(M^2)$. The complexities of PUMA and TPUMA algorithms are $\mr{O}(M^3)$ and $\mr{O}(k_t M(R+F-1))+\sum_{r=1}^R\mr{O}(K(F+1)M_r^3)$, respectively.  Compared to PUMA and TPUMA, the proposed algorithm has an attractive computational complexity for large size signals. \soul Figure~\ref{fig_cpu_time} shows the CPU time results  of the proposed, Tensor-ESPRIT and TPUMA algorithms versus $M_1$ for a $3$-D damped signal containing two modes with $M_2 = M_3 = 4$.  We observe that the proposed method has low computational complexity compared to TPUMA and Tensor-ESPRIT when $M_1$ is large. \fin

\begin{figure}[t]
\centerline{\includegraphics[width=\largeur]{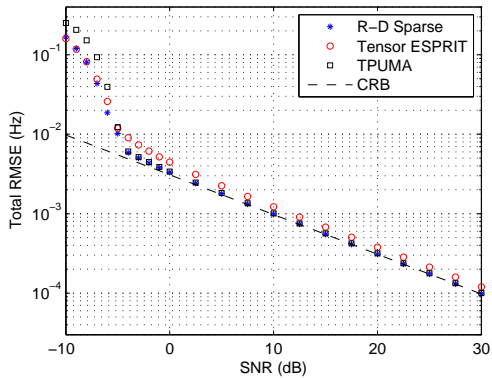}}
\caption{Frequency total root-mean square error for a 3-D signal containing two 3-D modes (Signal \#2). $(M_1,M_2,M_3) = (8,8,8). $ 1000 Monte-Carlo.}
\label{fig:two_3D_modes}
\end{figure}
\begin{figure}[ht]
\centerline{\includegraphics[width=\largeur]{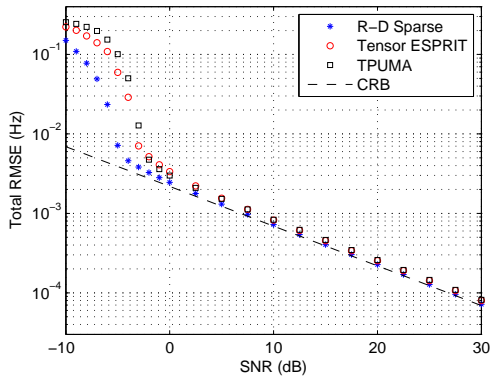}}
\caption{Frequency total root-mean square error for a 3-D signal containing 3 modes with identical modes in two dimensions (Signal \#3). $(M_1,M_2,M_3) = (10,10,10). $ 200 Monte-Carlo.}
\label{fig:3modes_Identical_In_towDimdension}
\end{figure}
\begin{figure}[ht]
\centerline{\includegraphics[width=\largeur]{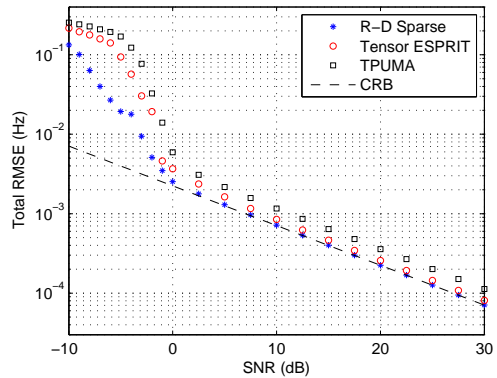}}
\caption{Frequency total root-mean square error for a 3-D signal containing 3 modes with identical modes in two dimensions (Signal \#4), same as Signal \#4 with close modes in the first dimension (0.28,0.12,0.2). $(M_1,M_2,M_3) = (10,10,10). $ 200 Monte-Carlo. }
\label{fig:3modes_Identical_In_TwoDimdension_Close}
\end{figure}

\begin{figure}
\centerline{\includegraphics[width=\largeur]{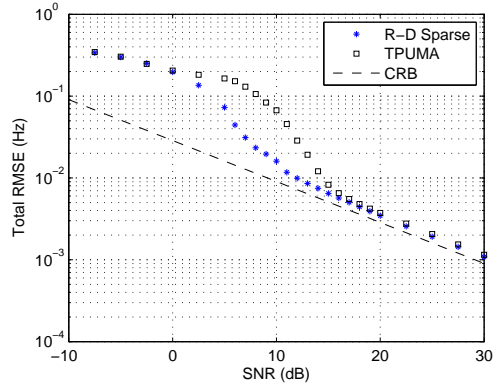}}
\caption {Frequency total root-mean square error for a 3-D signal containing 4 modes with $(M_1,M_2,M_3) = (10,3,3)$ (Signal \#5).  200 Monte-Carlo. }
\label{fig_signal_10x3x3}
\end{figure}

\begin{figure}[t]
\centerline{\includegraphics[width=\largeur]{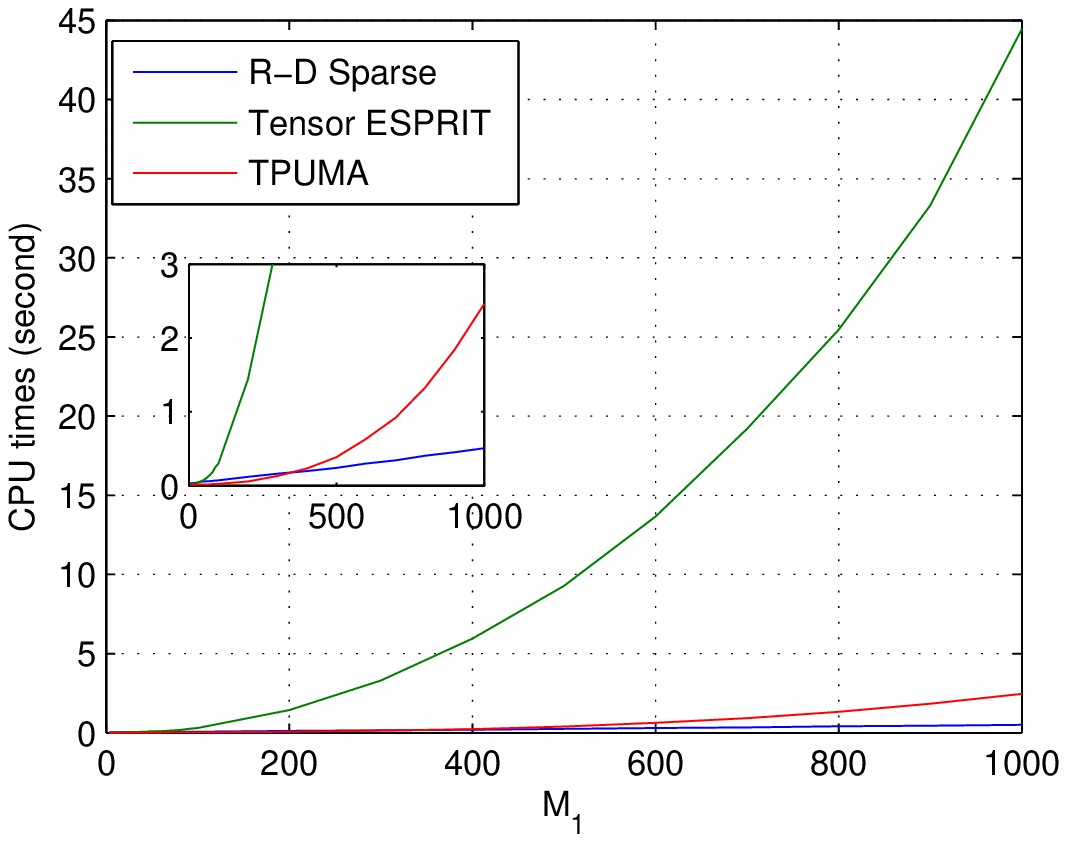}}
\caption {Average CPU time for a single run under $M_2 = M_3 = 4$ and $F = 2$.}\vspace{-3mm}
\label{fig_cpu_time}
\end{figure}

\section{Conclusion}\label{sec:conclusion}

We presented an efficient sparse estimation approach for the analysis of multidimensional ($R$-D) damped or undamped modal signals. The idea consists in exploiting the simultaneous sparse approximation principle to separate this joint estimation problem into $R$ multiple measurements problems. To be able to handle large size signals and yield accurate estimates, a multigrid dictionary refinement scheme  is associated with the simultaneous orthogonal matching pursuit (SOMP) algorithm. We gave the convergence proof of the  the refinement procedure in the single tone case. Then, for the general multiple tones $R$-D case, the signal tensor model is decomposed in order to handle each tone separately in an iterative scheme so that the pairing of the $R$-D parameters is automatically achieved. Also, the CRLB of the $R$-D modal signal parameters were derived. The tests performed on simulated signals showed that the proposed algorithm attains the CRLB and outperforms state-of-the-art subspace algorithms. We also have shown that the complexity of the algorithm is linear with respect to the number of measurements, which allows the processing of large size signals. Finally, it is worth mentioning that this approach can be straightforwardly applied to other multidimensional array processing problems.

\soul
\begin{appendices}
\section{SOMP Algorithm}\label{appendix_SOMP}
In this appendix we report the SOMP method (Algorithm~\ref{algo:s_omp})~\cite{Tropp:06}. In this description, $\mv{j}_{m_2}$ denotes the $m_2$\textsuperscript{th} vector of the canonical basis in $\mathbb{C}^{M_2}$.
\IncMargin{0mm}\begin{algorithm}[h!]
\footnotesize
\DontPrintSemicolon
\SetKwData{Left}{left}\SetKwData{This}{this}\SetKwData{Up}{up} \SetKwFunction{Union}{Union}
\SetKwFunction{FindCompress}{FindCompress} \SetKwInOut{Input}{input}\SetKwInOut{Output}{output}
\soul \Input {A matrix $\mv{Y} \in \mathbb{C}^{M_1\times M_2}$, a matrix $\mv{Q} \in \mathbb{C}^{M_1\times N}$ (with normalized columns)}
\Output{An index set $\Omega$ of activated atoms. A matrix of sparse vectors $\mv{X}\in \mathbb{C}^{N\times M_2}$}
\BlankLine
{\bf initialization}: $k=0,\Omega_0= \varnothing, \mv{X}= \mv{0}$, $\mv{R}_0=\mv{Y}$
\caption{SOMP}
\While{halting criterion false}{
	$k=k+1$\;
	 $n_k \in {\rm arg \ max}_n{\sum_{m_2=1}^{M_2}|\langle \mathbf{R}_{k-1}\mv{j}_{m_2}, \mv{q}_n \rangle|}$ \hspace{3mm} 
$\Omega_k=\Omega_{k-1}\cup\{n_k\}$ \; 
	$\mv{X}_k= (\mv{Q}_{\Omega_k}^{H}\mv{Q}_{\Omega_k})^{-1}\mv{Q}_{\Omega_k}^{H} \mv{Y}$ \hspace{3mm}\; 
	$\mathbf{R}_k = \mathbf{Y} - \mv{Q}_{\Omega_k}\mathbf{X}_k$\;
	}
	\Return{$\Omega= \Omega_k, \mv{X}=\mv{X}_k$} \fin
\label{algo:s_omp}
\end{algorithm}

\vspace{3mm}
\section{Proof of Theorem \ref{proposition_algoMTSM}} \label{appendix_proof_proposition_algoMTSM}
\soul
We begin the proof by introducing the following \red lemma. \soul
\begin{lemma}\label{pro_perturbation}
Consider $\widetilde{\ts{Y}} = \ts{Y} + \Delta{\ts{Y}}$, where $\widetilde{\ts{Y}}$ is the perturbed version of the data tensor $ \ts{Y}$ and $\Delta{\ts{Y}}$ is the perturbation. 
 Assuming that $\Delta\ts{Y}$ is sufficiently small such that the \red ordering of the $F$ \soul singular values in $\mSigm$ in (\ref{eq_svd_Y1}) is the same as the \red ordering \soul of the corresponding singular values when $\Delta\ts{Y}=0$. Then the perturbation $\Delta\ts{Y}_f$ contains a linear combination of all $\ts{Y}_f, f= 1,\ldots,F$:
 \begin{equation*}
  \Delta\ts{Y}_f  = \ts{D}_f  + \sum_{i=1}^F v_{f,i} \ts{Y}_{i}
\end{equation*}
 where $\vv_f^\t = [v_{f,1},\ldots,v_{f,F}] = \Delta\mA_{1}^\dag(f,:)\mA_1$ and $\ts{D}_f = \Delta\ts{Y}  \con_1 \va_{f,1}\mA_{1}^\dag(f,:) +\ts{Y}_{s,f} \con_1 \Delta\va_{f,1}$. 
\end{lemma}
\begin{proof}
From (\ref{eq:Yshat}) $\red \ts{S} \soul = \ts{Y}\con_1 \mA_1^\dag$, we differentiate and obtain 
\begin{align*} 
\red \Delta\ts{S} \soul & = \Delta\ts{Y}\con_1 \mA_1^\dag + \ts{Y}\con_1 \Delta\mA_1^\dag 
\end{align*}
Then,
\begin{align*} 
 \red \Delta\ts{S}_{f} \soul & =  \Delta\ts{Y}\con_1 \mA_{1}^\dag(f,:) + \red \ts{S} \soul \con_1 \underbrace{\Delta\mA_{1}^\dag(f,:)\mA_1}_{\vv^\t}  \\
& = \Delta\ts{Y}\con_1 \mA_{1}^\dag(f,:) + \red \ts{S} \soul \con_1 \vv^\t \\
& = \Delta\ts{Y}\con_1 \mA_{1}^\dag(f,:) + \sum_{p=1}^F v_{f,p} \red \ts{S}_{p} 
\end{align*}
$ \ts{Y}_f$ is estimated using $ \ts{Y}_f = \red \ts{S}_{f} \soul \con_1 \va_{f,1}$, we differentiate and obtain 
\begin{align*}
  \Delta\ts{Y}_f & = \red \Delta\ts{S}_{f} \soul \con_1 \va_{f,1} + \red \ts{S}_{f} \soul \con_1 \Delta\va_{f,1} \\ 
  & = \sum_{p=1}^F v_{f,p} \ts{Y}_{p} +  \red \ts{S}_{f} \soul \con_1 \Delta\va_{f,1}  + \Delta\ts{Y}  \con_1 \va_{f,1}\mA_{1}^\dag(f,:)
\end{align*} 
\hfill $\blacksquare$
\end{proof}
Using the previous lemma

\begin{align*} \bar{\ts{Y}}_{f}^{(0)}  &= \red \ts{S}_{f} \soul \con_{1}(\va_{f,1}+v_{f,f}\va_{f,1}+\Delta\va_f)+\\
 &  \quad \quad \sum_{p=1,p\neq f}^F v_{f,p} \ts{Y}_{p}+\Delta\ts{Y}  \con_1 \va_{f,1}\mA_{1}^\dag(f,:) 
\end{align*}
Therefore,  $\va_{f,2},\dots,\va_{f,R}, f = 1,\ldots,F$ can be estimated using STSM algorithm since 
\begin{align*} \bar{\mY}_{{f}_{(r)}}^{(0)}  &= c_f \va_{f,r}(\va_{f,R}\kron\cdots \kron \va_{f,r+1}\kron\va_{f,r-1}\kron \\ 
& \quad \quad  \cdots \kron (\va_{f,1}+v_{f,f}\va_{f,1}+\Delta\va_{f,1}))\\
 &  \quad \quad + \left(\sum_{p=1,p\neq f}^F v_{f,p} \ts{Y}_{p}+\Delta\ts{Y}  \con_1 \va_{f,1}\mA_{1}^\dag(f,:)\right)_{(r)}. 
\end{align*}
Since $\bar{\mY}_{{f}_{(1)}}^{(0)}$ has the following form  
\begin{align*} \bar{\mY}_{{f}_{(1)}}^{(0)}  &= c_f (\va_{f,1}+v_{f,f}\va_{f,1}+\Delta\va_{f,1}) (\va_{f,R}\kron\cdots \kron \va_{f,2}) \\
 &  \quad \quad  + \left(\sum_{p=1,p\neq f}^F v_{f,p} \ts{Y}_{p}+\Delta\ts{Y}  \con_1 \va_{f,1}\mA_{1}^\dag(f,:)\right)_{(1)}, 
\end{align*}
we estimate $c_f \va_{f,1}$ by least squares once $\va_{f,2},\dots,\va_{f,R}$ are estimated using STSM 
\begin{align*}  \red \widehat{{c}_f {\va}_{f,1}} \soul & = \min_{\va} \|  \bar{\ts{Y}}_{f}^{(0)} - \va \out \hat{\va}_{f,2}\out \cdots \out \hat{\va}_{f,R}  \| \\
& =  \bar{\mY}_{{f}_{(1)}}^{(0)}  \left( (\hat{\va}_{f,R}\kron \cdots \kron \hat{\va}_{f,2})^\t \right)^\dag
\end{align*}
So, we put $\hat{\ts{Y}}_{f}^{(0)}  =  \red \widehat{{c}_f{\va}_{f,1}} \soul \out \hat{\va}_{f,2}\out \cdots \out \hat{\va}_{f,R} $ and $\ts{R}_f = \bar{\ts{Y}}_{f}^{(0)} - \hat{\ts{Y}}_{f}^{(0)}. $
Therefore, the procedure to  estimate ${\ts{Y}}_{f}$ at iteration $i=0,\ldots,K$ can be summarized  in \soul ~(\ref{eq_1proce_df1_afR}), (\ref{eq_2proce_df1_afR}) and (\ref{eq_3proce_df1_afR}). 
\soul Note that this procedure is optimal because STSM and the least squares are optimal when they are used to estimate  $\va_{f,2},\dots,\va_{f,R}, f = 1,\ldots,F$ and  $c_f \va_{f,1}, f=1,\ldots,F$, respectively.  \soul

Now we present the technique for improving the estimation of ${\ts{Y}}_{f} $. Let $ \ts{R}_F^{(0)}  =  \ts{R}_0^{(1)} = \red\tilde{\ts{Y}} - \sum_{f=1}^F \hat{\ts{Y}}_{f}^{(0)} $ and 
\begin{align}
 \hat{\ts{Y}}_f^{(1)} & = \argmin{\ts{X}\in \red \mathcal{H} \soul}{ \|\hat{\ts{Y}}_f^{(0)}  + \ts{R}_{f-1}^{(1)}  - \ts{X} \| }
 \label{eq_minimum_tsY_f}
\end{align}
where $\ts{R}_{f}^{(1)} = \hat{\ts{Y}}_f^{(0)}  + \ts{R}_{f-1}^{(1)}-   \hat{\ts{Y}}_f^{(1)}, f= 1,\ldots,F$, and $\hat{\ts{Y}}_f^{(i)}$ is an improved estimate of ${\ts{Y}}_f$. We follow the same procedure as described in equations~(\ref{eq_1proce_df1_afR}), (\ref{eq_2proce_df1_afR}) and (\ref{eq_3proce_df1_afR}) to calculate $\hat{\ts{Y}}_f^{(1)}$. 


We can state that there is improvement in the estimation of $\ts{Y}_f$ if  \red
\begin{align}\left\| \tilde{\ts{Y}} - \sum_{f=1}^F  \hat{\ts{Y}}_f^{(1)} \right\| \leq \left\|\tilde{\ts{Y}} - \sum_{f=1}^F  \hat{\ts{Y}}_f^{(0)} \right\|
\label{eq_RF0_RF1}
\end{align} \soul
We have $\|\ts{R}_0^{(1)} \| = \| \bar{\ts{Y}}_{1}^{(0)} - \hat{\ts{Y}}_{1}^{(0)} + \sum_{f=2}^F \ts{R}_f +\red \ts{V} \soul \|$  \red where $\ts{V} = \tilde{\ts{Y}} - \bar{\ts{Y}} $ and $ \bar{\ts{Y}}  = \sum_{f=1}^F \bar{\ts{Y}}_{f}^{(0)}$. It can be verified that \soul 
\small 
$$\|\ts{R}_f^{(1)} \| = \left\| \left(\bar{\ts{Y}}_f^{(0)}  + \sum_{p=1}^{f-1} (\bar{\ts{Y}}_p^{(0)}- \hat{\ts{Y}}_p^{(1)}) +  \sum_{p=f+1}^{F} \ts{R}_p + \red \ts{V} \soul \right) -  \hat{\ts{Y}}_f^{(1)}\right\| $$
$$\|\ts{R}_{f-1}^{(1)} \| = \left\| \left(\bar{\ts{Y}}_f^{(0)}  + \sum_{p=1}^{f-1} (\bar{\ts{Y}}_p^{(0)}- \hat{\ts{Y}}_p^{(1)}) +  \sum_{p=f+1}^{F} \ts{R}_p + \red \ts{V} \soul \right) -  \hat{\ts{Y}}_f^{(0)}\right\| $$
\normalsize
\red However, \soul from equation~(\ref{eq_minimum_tsY_f}), $\hat{\ts{Y}}_f^{(1)}$ is the minimizer with respect to $\ts{X} \red \in \mathcal{H} \soul $ of 
$$  \left\| \left(\bar{\ts{Y}}_f^{(0)}  + \sum_{p=1}^{f-1} (\bar{\ts{Y}}_p^{(0)}- \hat{\ts{Y}}_p^{(1)}) +  \sum_{p=f+1}^{F} \ts{R}_p \right) -  {\ts{X}}\right\| $$
Therefore, $\|\ts{R}_f^{(1)} \| \leq \|\ts{R}_{f-1}^{(1)} \| , f= 1,\ldots,F$. As consequence, $\|\ts{R}_F^{(1)} \| \leq \|\ts{R}_{F}^{(0)} \|$, which we are seeking in expression~(\ref{eq_RF0_RF1}). Similarly, we can prove that \red $\|\ts{R}_F^{(i)} \| \leq \|\ts{R}_{F}^{(i-1)} \|, i > 1,$ using the general forms of $\ts{R}_f^{(i)}$ and $\ts{R}_{f-1}^{(i)}$ 
\small
\begin{align}
\ts{R}_f^{(i)} & = \left(\bar{\ts{Y}}_f^{(0)}  + \sum_{p=1}^{f-1} (\bar{\ts{Y}}_p^{(0)}- \hat{\ts{Y}}_p^{(i)}) +  \sum_{p=f+1}^{f-1} (\bar{\ts{Y}}_p^{(0)}- \hat{\ts{Y}}_p^{(i-1)}) +  \ts{V} \right) \nonumber \\ 
& \quad \quad  -  \hat{\ts{Y}}_f^{(i)}
\end{align}
\begin{align}
\ts{R}_{f-1}^{(i)} & = \left(\bar{\ts{Y}}_f^{(0)}  + \sum_{p=1}^{f-1} (\bar{\ts{Y}}_p^{(0)}- \hat{\ts{Y}}_p^{(i)}) +  \sum_{p=f+1}^{f-1} (\bar{\ts{Y}}_p^{(0)}- \hat{\ts{Y}}_p^{(i-1)}) +  \ts{V} \right) \nonumber \\ 
& \quad \quad  -  \hat{\ts{Y}}_f^{(i-1)}
\end{align}\normalsize
which we are seeking in (\ref{eq_RFim1_RFi}). 
\fin

\end{appendices}

\fin

\section{Acknowledgment }
The authors would like to thank PhD Weize Sun for providing them with the TPUMA algorithm code.


\begin{thebibliography}{10}

\bibitem{Stoica:97}
P.~Stoica and R.~Moses, \emph{Introduction to Spectral Analysis}.\hskip 1em
  plus 0.5em minus 0.4em\relax Upper Saddle River, NJ: Prentice Hall, 1997.

\bibitem{gershman2005space}
A.~B. Gershman and N.~D. Sidiropoulos, \emph{Space-time processing for MIMO
  communications}.\hskip 1em plus 0.5em minus 0.4em\relax Wiley Online Library,
  2005.

\bibitem{Nion:10}
D.~Nion and S.~Sidiropoulos, ``Tensor algebra and multidimensional retrieval in
  signal processing for {MIMO} radar,'' \emph{IEEE Trans. Signal Process.},
  vol.~58, no.~1, pp. 5693--5705, 2010.

\bibitem{Sacchini93}
J.~Sacchini, W.~Steedly, and R.~Moses, ``Two-dimensional {P}rony modeling and
  parameter estimation,'' \emph{IEEE Trans. Signal Process.}, vol.~41, no.~11,
  pp. 3127--3137, 1993.

\bibitem{Hua92}
Y.~Hua, ``Estimating two-dimensional frequencies by matrix enhancement and
  matrix pencil,'' \emph{IEEE Trans. Signal Process.}, vol.~40, no.~9, pp.
  2267--2280, 1992.

\bibitem{Rouquette2001}
S.~Rouquette and M.~Najim, ``Estimation of frequencies and damping factors by
  two-dimensional {ESPRIT} type methods,'' \emph{IEEE Trans. Signal Process.},
  vol.~49, no.~1, pp. 237--245, 2001.

\bibitem{mokios20043d}
K.~Mokios, N.~Sidiropoulos, M.~Pesavento, and C.~Mecklenbrauker, ``On {3-D}
  harmonic retrieval for wireless channel sounding,'' in \emph{Proc. IEEE
  ICASSP}, Montreal, Canada, May 2004, pp. ii89--ii92.

\bibitem{Lui2006}
J.~Liu and X.~Liu, ``An eigenvector-based approach for multidimensional
  frequency estimation with improved identifiability,'' \emph{IEEE Trans.
  Signal Process.}, vol.~54, no.~12, pp. 4543--4556, December 2006.

\bibitem{Lui07_november}
J.~Liu, X.~Liu, and X.~Ma, ``Multidimensional frequency estimation with finite
  snapshots in the presence of identical frequencies,'' \emph{IEEE Trans.
  Signal Process.}, vol.~55, pp. 5179--5194, 2007.

\bibitem{haardt2008higher}
M.~Haardt, F.~Roemer, and G.~Del~Galdo, ``Higher-order {SVD}-based subspace
  estimation to improve the parameter estimation accuracy in multidimensional
  harmonic retrieval problems,'' \emph{IEEE Trans. Signal Process.}, vol.~56,
  no.~7, pp. 3198--3213, 2008.

\bibitem{sun2012accurate}
W.~Sun and H.~C. So, ``Accurate and computationally efficient tensor-based
  subspace approach for multi-dimensional harmonic retrieval,'' \emph{IEEE
  Trans. Signal Process.}, vol.~60, no.~10, pp. 5077--5088, October 2012.

\bibitem{so2010efficient}
H.~So, F.~Chan, W.~Lau, and C.~Chan, ``An efficient approach for
  two-dimensional parameter estimation of a single-tone,'' \emph{IEEE Trans.
  Signal Process.}, vol.~58, no.~4, pp. 1999--2009, April 2010.

\bibitem{huang2012multidimensional}
L.~Huang, Y.~Wu, H.~So, Y.~Zhang, and L.~Huang, ``Multidimensional sinusoidal
  frequency estimation using subspace and projection separation approaches,''
  \emph{IEEE Trans. Signal Process.}, vol.~60, no.~10, pp. 5536--5543, 2012.

\bibitem{lin2013efficient}
C.~Lin and W.~Fang, ``Efficient multidimensional harmonic retrieval: A
  hierarchical signal separation framework,'' \emph{IEEE Signal Process.
  Lett.}, vol.~20, no.~5, pp. 427--430, May 2013.

\bibitem{Sahnoun:2012Eusipco}
S.~Sahnoun, E.-H. Djermoune, and D.~Brie, ``Sparse multigrid modal estimation:
  Initial grid selection,'' in \emph{Proc. European Signal Process. Conf.
  (EUSIPCO-2012)}, 2012, pp. 450--454.

\bibitem{goodwin1999matching}
M.~Goodwin and M.~Vetterli, ``Matching pursuit and atomic signal models based
  on recursive filter banks,'' \emph{IEEE Trans. Signal Process.}, vol.~47,
  no.~7, pp. 1890--1902, 1999.

\bibitem{Malioutov2005}
D.~Malioutov, M.~Cetin, and A.~Willsky, ``A sparse signal recontruction
  perspective for source localization with sensor arrays,'' \emph{IEEE Trans.
  Signal Process.}, vol.~53, no.~8, pp. 3010--3022, 2005.

\bibitem{Stoica:2011spice1}
P.~Stoica, P.~Babu, and J.~Li, ``{SPICE}: A sparse covariance-based estimation
  method for array processing,'' \emph{IEEE Trans. Signal Process.}, vol.~59,
  no.~2, pp. 629--638, 2011.

\bibitem{sahnoun2012_Mltigrid}
S.~Sahnoun, E.-H. Djermoune, C.~Soussen, and D.~Brie, ``Sparse multidimensional
  modal analysis using a multigrid dictionary refinement,'' \emph{EURASIP J.
  Adv. Signal Process.}, March 2012. 

\bibitem{Sahnoun:13}
S.~Sahnoun, E.-H. Djermoune, and D.~Brie, ``Sparse modal estimation of {2-D
  NMR} signals,'' in \emph{Proc. IEEE ICASSP}, Vancouver, Canada, May 2013, pp.
  8751--8755. 

\bibitem{sward2014high}
J.~Sward, S.~I. Adalbjornsson, and A.~Jakobsson, ``High resolution sparse
  estimation of exponentially decaying signals,'' in \emph{2014 IEEE
  International Conference on Acoustics, Speech and Signal Processing
  (ICASSP)}.\hskip 1em plus 0.5em minus 0.4em\relax IEEE, 2014, pp. 7203--7207.

\bibitem{adalbjornsson2014high}
S.~I. Adalbj{\"o}rnsson, J.~Sw{\"a}rd, and A.~Jakobsson, ``High resolution
  sparse estimation of exponentially decaying two-dimensional signals,'' in
  \emph{22nd European Signal Processing Conference-EUSIPCO 2014}.\hskip 1em
  plus 0.5em minus 0.4em\relax EURASIP, 2014.

\bibitem{Tropp:06}
J.~Tropp, A.~Gilbert, and M.~Strauss, ``Algorithms for simultaneous sparse
  approximation. {P}art {I}: Greedy pursuit,'' \emph{Signal Process.}, vol.~86,
  pp. 572--588, 2006.

\bibitem{boyer2008deterministic}
R.~Boyer, ``Deterministic asymptotic cramer--rao bound for the multidimensional
  harmonic model,'' \emph{Signal Processing}, vol.~88, no.~12, pp. 2869--2877,
  2008.

\bibitem{Clark1994}
M.~Clark and L.~Scharf, ``Two-dimensional modal anlysis based on maximum
  likelihood,'' \emph{IEEE Trans. Signal Process.}, vol.~42, no.~6, pp.
  1443--1451, 1994.

\bibitem{Como14:spmag}
P.~Comon, ``Tensors: a brief introduction,'' \emph{IEEE Sig. Proc. Magazine},
  vol.~31, no.~3, May 2014, special issue on BSS. hal-00923279.

\bibitem{Kolda:2009tensor}
T.~Kolda and B.~Bader, ``Tensor decompositions and applications,'' \emph{SIAM
  Review}, vol.~51, no.~3, pp. 455--500, 2009.

\bibitem{Rakotomamonjy2011surveying}
A.~Rakotomamonjy, ``Surveying and comparing simultaneous sparse approximation
  (or group-{Lasso}) algorithms,'' \emph{Signal Process.}, vol.~91, no.~7, pp.
  1505--1526, 2011.

\bibitem{Tropp:04}
J.~Tropp, ``Greed is good: algorithmic results for sparse approximation,''
  \emph{IEEE Trans. Inf. Theory}, vol.~50, pp. 2231--2242, Oct 2004.

\bibitem{Tropp:16}
J.A.Tropp, ``Algorithms for simultaneous sparse approximation. {P}art {II}:
  Convex relaxation,'' \emph{Signal Process.}, vol.~86, pp. 589--602, March
  2006.

\bibitem{Mishali:08}
M.~Mishali and Y.~Eldar, ``Reduce and boost: Recovering arbitrary sets jointly
  sparse vectors,'' \emph{IEEE Trans. Signal Process.}, vol.~56, pp.
  4692--4702, Oct 2008.

\bibitem{Lee:12}
K.~Lee, Y.~Bresler, and M.~Junge, ``Subspace methods for joint sparse
  recovery,'' \emph{IEEE Trans. Inf. Theory}, vol.~58, pp. 3613--3641, Jun
  2012.

\bibitem{jiang2001almost}
T.~Jiang, N.~D. Sidiropoulos, and J.~M. Ten~Berge, ``Almost-sure
  identifiability of multidimensional harmonic retrieval,'' \emph{Signal
  Processing, IEEE Transactions on}, vol.~49, no.~9, pp. 1849--1859, 2001.

\bibitem{Yao:95}
Y.-X. Yao and S.~Pandit, ``Cram\'er-{R}ao lower bounds for a damped sinusoidal
  process,'' \emph{IEEE Trans. Signal Process.}, vol.~43, no.~4, pp. 878--885,
  1995.

\bibitem{Kay:93}
S.~Kay, \emph{Fundamentals of statistical signal processing. Estimation
  theory}.\hskip 1em plus 0.5em minus 0.4em\relax Englewood Cliffs, New Jersey:
  Prentice Hall International Editions, 1993.

\bibitem{Tropp2010computational}
J.~Tropp and S.~Wright, ``Computational methods for sparse solution of linear
  inverse problems,'' \emph{Proceedings of the {IEEE}}, vol.~98, no.~6, pp.
  948--958, 2010.

\bibitem{Lemma:03}
A.~Lemma, A.-J. van~der Even, and E.~Deprettere, ``Analysis of joint
  angle-frequency estimation using {ESPRIT},'' \emph{IEEE Trans. Signal
  Process.}, vol.~51, pp. 1264--1283, May 2003.

\bibitem{Djermoune:2009perturbation}
E.-H. Djermoune and M.~Tomczak, ``Perturbation analysis of subspace-based
  methods in estimating a damped complex exponential,'' \emph{IEEE Trans.
  Signal Process.}, vol.~57, no.~11, pp. 4558--4563, 2009.

\bibitem{Hua1990}
Y.~Hua and T.~Sarkar, ``Matrix pencil method for estimating parameters of
  exponentially damped/undamped sinusoids in noise,'' \emph{IEEE Trans. Acoust.
  Speech Signal Process.}, vol.~38, pp. 814--824, 1990.

\end{thebibliography}

\end{document}